\documentclass[a4paper,border=2px,12pt,accepted=2024-12-03]{quantumarticle}
\pdfoutput=1
\usepackage[scale=0.8]{geometry}
\usepackage{amsmath,amssymb,amsthm,amsfonts}
\usepackage{graphicx}

\usepackage{tikz}
\usetikzlibrary{quantikz}
\usepackage{graphicx}
\usepackage{color}
\usepackage{float}
\usepackage{hyperref}
\DeclareMathOperator{\diag}{diag}
\usepackage{authblk}

\newtheorem{mythm}{Theorem}

\newtheorem{mylemma}{Lemma}
\newtheorem{myremark}{Remark}
\numberwithin{equation}{section}
\graphicspath{{./fig/}}

\title{Quantum circuits for partial differential equations via Schr\"odingerisation}

\author[1]{Junpeng Hu}
\thanks{hjp3268@sjtu.edu.cn}
\author[1,2,4]{Shi Jin}
\thanks{shijin-m@sjtu.edu.cn}
\author[2,3,4]{Nana Liu}
\thanks{nana.liu@quantumlah.org}
\author[1,2,4]{Lei Zhang}
\thanks{lzhang2012@sjtu.edu.cn}
\affil[1]{\small School of Mathematical Sciences, Shanghai Jiao Tong University, Shanghai, 200240, China}
\affil[2]{Institute of Natural Sciences, Shanghai Jiao Tong University, Shanghai, 200240, China}
\affil[3]{University of Michigan-Shanghai Jiao Tong University Joint Institute, Shanghai, 200240, China}
\affil[4]{MOE-LSC, Shanghai Jiao Tong University, Shanghai, 200240, China}

\begin{document}

\maketitle

\begin{abstract}
Quantum computing has emerged as a promising avenue for achieving significant speedup, particularly in large-scale PDE simulations, compared to classical computing. One of the main quantum approaches involves utilizing Hamiltonian simulation, which is directly applicable only to Schr\"odinger-type equations. To address this limitation, Schr\"odingerisation techniques have been developed, employing the warped transformation to convert general linear PDEs into Schr\"odinger-type equations. However, despite the development of Schr\"odingerisation techniques, the explicit implementation of the corresponding quantum circuit for solving general PDEs remains to be designed. In this paper, we present a detailed implementation of a quantum algorithm for general PDEs using Schr\"odingerisation techniques. We provide examples of the heat equation, and the advection equation approximated by the upwind scheme, to demonstrate the effectiveness of our approach. Complexity analysis is also carried out to demonstrate the quantum advantages of these algorithms in high dimensions over their classical counterparts. Several numerical experiments demonstrate the validity of these proposed quantum circuits.
\end{abstract}

\section{Introduction}

Partial differential equations (PDEs) models are essential tools to investigate dynamic behaviors of many physical systems, arise in applications such as reservoir modeling, atmospheric and ocean circulation, and high-frequency scattering. Simulating complex dynamical systems within feasible computational time is of utmost importance. Despite the significant progress achieved through the utilization of high-performance computing clusters, the computational overhead remains a major bottleneck, particularly for large systems, high dimensional and multiscale problems.

In recent years, quantum computing \cite{feynman2018simulating,nielsen2010quantum} has emerged as a promising avenue for achieving significantly faster computation compared to classical computing. Despite the current limitations in hardware scalability and noise resistance, there has been remarkable progress in developing quantum algorithms for scientific and engineering computing, which will be useful when quantum computers become available one day, or can also help to guide the design of analog quantum computing \cite{CVPDE2023}. Among many potential applications, one particularly promising area is the utilization of quantum computers as solvers for both quantum and classical partial differential equations (PDEs).

Lloyd's groundbreaking work \cite{lloyd1996universal} introduced the first explicit quantum simulation algorithm, which allowed for the simulation of Hamiltonians containing local interaction terms. Building upon this, Aharonov and Ta-Shma \cite{aharonov2003adiabatic} proposed an efficient simulation algorithm capable of handling a broader class of sparse Hamiltonians. Subsequent research endeavors \cite{childs2004quantum,berry2007efficient,wiebe2011simulating,childs2010relationship,berry2009black,childs2012hamiltonian,berry2015simulating,berry2014exponential,berry2015hamiltonian,berry2020time,zhao2022hamiltonian} have further refined and improved these simulation techniques. In their work \cite{berry2015hamiltonian}, Berry et al. combined the quantum walk approach \cite{childs2010relationship,berry2009black} with the fractional-query approach \cite{berry2014exponential}. This innovative fusion resulted in a significant reduction in query complexity, achieving an algorithmic complexity of $O\left( \tau \frac{\log(\tau/\varepsilon)}{\log\log(\tau/\varepsilon)} \right)$, where $\tau := s \left| H \right|_{\max}t$ and $\varepsilon$ is the desired precision. Notably, this approach exhibited a near-linear dependence on both the sparsity $s$ and the evolution time $t$, while offering exponential speedup over $\varepsilon$. Moreover, Berry et al. established a lower bound, demonstrating the near-optimality of this result with respect to scaling in either $\tau$ or $\varepsilon$ independently.

Achieving quantum speedup for ordinary or partial differential equations (ODEs/PDEs) that are not of the Schr\"odinger equation type is of significant importance in scientific and engineering applications. A commonly employed strategy is to discretize the spatial and temporal domains, thereby transforming linear partial differential equations (PDEs) into a system of linear algebraic equations. Subsequently, quantum algorithms are utilized to solve this transformed system \cite{harrow2009quantum,cao2012quantum,childs2017quantum}. The application of quantum algorithms in solving PDEs has the potential to yield polynomial or super-polynomial speedups \cite{berry2014high,berry2017quantum,childs2021high}. Jin et al. conducted a comprehensive study on the time complexity of quantum difference methods for linear and nonlinear high-dimensional and multiscale PDEs within the framework of Asymptotic-Preserving schemes \cite{jin2022time,jin2022quantumobservables,jin2022timenonlinear}. In their work \cite{jin2022quantum,jin2022quantumdetail}, they proposed a recent approach known as ``Schr\"odingerisation'' for solving general linear PDEs. This approach involves converting linear PDEs into a system of Schr\"odinger equations in one higher dimension, using a simple transformation called the warped phase transformation. ``Schr\"odingerisation'' enables quantum simulations for general linear ordinary differential equations (ODEs) and PDEs \cite{lu2023quantum}, and iterative  linear algebra solvers \cite{JL-LA}.  An alternative approach is presented in \cite{an2023linear}.

While significant progress has been made in the field, the explicit design of quantum circuits for solving partial differential equations (PDEs) remains open. A recent study by Sato et al. \cite{sato2024hamiltonian} has introduced a novel approach that utilizes the Bell basis to construct scalable quantum circuits specifically for wave or Schr\"odinger-type PDEs. In our paper, we construct quantum circuits for general PDEs--which not necessarily follow unitary dynamics --by using the Schr\"odingerisation technique. 

\paragraph{Contribution} We present a detailed circuit implementation of the approximating unitary matrix $V$, as depicted in Figure \ref{fig:circuit:schro}. This explicit  implementation has not been reported in previous literature. To demonstrate the effectiveness of our approach, we provide two examples: the heat equation, and the advection equation with the upwind scheme which was not studied in \cite{sato2024hamiltonian}. Complexity analysis is performed to demonstrate the quantum advantages of these algorithms in high dimensions over their classical counterparts. Numerical experiments for both the heat equation and the advection equation were performed using the {\it Qiskit} package \cite{Qiskit}. The results, obtained either through statevectors or observations from the quantum circuits, closely align with those derived from classical computational methods.

\paragraph{Organization} The paper is structured as follows: In Section \ref{sec:fdo}, we introduce the basic notations. The quantum circuit implementation for the heat equation is introduced in Section \ref{sec:heat}, followed by the implementation of the upwind scheme for the advection equation in Section \ref{sec:advection}. To showcase the scalability of the proposed quantum circuits, we conduct a complexity analysis in Section \ref{sec:complexity}. Several numerical experiments are presented to demonstrate the validity of these proposed quantum circuits in Section \ref{sec:numerical}. Finally, we present our conclusions in Section \ref{sec:conclusion}.

\paragraph{Notation} In the context of complexity analysis, we use the big $O$ notation to describe the upper bound, and use $\tilde{O}$ to denote $O$ where logarithmic terms are ignored. The symbols $X,Y,Z$ represent the Pauli matrices.

\section{Representation of finite difference operator}
\label{sec:fdo}

We begin by considering the one-dimensional case in order to establish the definition of the finite difference operator and its representation in quantum states.

Given a one-dimensional domain $\Omega:=[0,L]$, which is uniformly subdivided into $N_x = 2^{n_x}$ intervals of length $h=\frac{L}{N_x}$, one can express a discrete function $u$ with $u_i$ defined on $x_i = ih$, for $i=0, \cdots, N_x-1$, as a quantum state $\ket{\mathbf{u}} = \mathbf{u} / \left\| \mathbf{u} \right\|, \mathbf{u} := \sum_{j=0}^{2^{n_x}-1} u_j |j\rangle$.

The shift operators can be defined as follows:
\begin{equation}
    \begin{aligned}
        S^{-} &:= \sum_{j=1}^{2^{n_x}-1} | j-1 \rangle \langle j | \\
        &= \sum_{j=1}^{n_x} I^{\otimes (n_x -j)} \otimes \sigma_{01} \otimes \sigma_{10}^{\otimes (j-1)} \\
        &\triangleq \sum_{j=1}^{n_x} s_j^{-}, \\
        S^{+} &:= (S^{-})^{\dagger} = \sum_{j=1}^{2^{n_x}-1} | j \rangle \langle j-1 | \\
        &= \sum_{j=1}^{n_x} I^{\otimes (n_x -j)} \otimes \sigma_{10} \otimes \sigma_{01}^{\otimes (j-1)} \\
        &\triangleq \sum_{j=1}^{n_x} s_j^{+}.
    \end{aligned}
\end{equation}
with basic $2\times 2$ matrices
\begin{equation}
\begin{aligned}
    \sigma_{01} &:= |0\rangle \langle 1|, \quad
    \sigma_{10} := |1\rangle \langle 0|, \\
    \sigma_{00} &:= |0\rangle \langle 0|, \quad
    \sigma_{11} := |1\rangle \langle 1|.
\end{aligned}
\end{equation}

The forward and backward finite difference operators are, 
\begin{equation}
        (D^{+} \mathbf{u})_j = \frac{u_{j+1} - u_{j}}{h}, \quad (D^{-} \mathbf{u})_j = \frac{u_{j} - u_{j-1}}{h},
\end{equation}
for $j = 0,1,\dots, N_x-1$.

Thus, the forward ($D_{P}^+$), backward ($D_{P}^-$), central ($D_{P}^{\pm}$) finite difference operators with respect to the first order derivative, and the central finite difference operator ($D_{P}^{\Delta}$) with respect to the second order derivative, combined with the periodic boundary conditions, can be expressed as follows,
\begin{equation}
    \begin{aligned}
        D_{P}^{+} &= \frac{S^{-} - I^{\otimes n_x} + \sigma_{10}^{\otimes n_x}}{h}, \\
        D_{P}^{-} &= \frac{I^{\otimes n_x} - S^{+} - \sigma_{01}^{\otimes n_x}}{h}, \\
        D_{P}^{\pm} &= \frac{S^{-} - S^{+} - \sigma_{01}^{\otimes n_x} + \sigma_{10}^{\otimes n_x}}{2h}, \\
        D_{P}^{\Delta} &= \frac{S^{-} + S^{+} - 2 I^{\otimes n_x} + \sigma_{01}^{\otimes n_x} + \sigma_{10}^{\otimes n_x}}{h^2}.
    \end{aligned}
\end{equation}

Similarly, the difference operator with Dirichlet boundary conditions for the left or right side can be expressed as
\begin{equation}
\begin{aligned}
    D_{D}^{+} &= \frac{S^{-} - I^{\otimes n_x}}{h}, \quad \text{(right side)}, \\
    D_{D}^{-} &= \frac{I^{\otimes n_x} - S^{+}}{h}, \quad \text{(left side)}.
\end{aligned}
\end{equation}
Considering the Dirichlet boundary conditions for both sides, i.e., $u_0 = u_{N_x} = 0$, it suffices to solve the numerical solution $\mathbf{u}:=[u_1;\cdots,u_{N_x-1}]$. In this case, we choose $N_x = 2^{n_x}+1$ and define $|\mathbf{u}\rangle = \mathbf{u} / \left\| \mathbf{u} \right\|, \mathbf{u} := \sum_{j=0}^{2^{n_x}-1} u_{j+1} |j\rangle$. Therefore, the difference operators with Dirichlet boundary conditions for both sides are 
\begin{equation}
\begin{aligned}
    D_{D}^{+} &= \frac{S^{-} - I^{\otimes n_x}}{h}, \\
    D_{D}^{-} &= \frac{I^{\otimes n_x} - S^{+}}{h}, \\
    D_{D}^{\pm} &= \frac{S^{-} - S^{+}}{2h}, \\
    D_{D}^{\Delta} &= \frac{S^{-} + S^{+} - 2 I^{\otimes n_x}}{h^2}.
\end{aligned}
\end{equation}

We will employ the following techniques described in \cite{sato2024hamiltonian}, which will be utilized multiple times throughout the paper. 
\begin{mylemma}{{\cite{sato2024hamiltonian}}}
\label{lemma:changebasis}
    Given an operator of the form $S = \ket{a}\bra{b} + \ket{b}\bra{a}$ with $\ket{a}$, $\ket{b} \in \mathbb{C}^{2^n}$ and $\langle a\ket{b}=0$, it can be decomposed into
    \begin{equation}
    \begin{aligned}
        S &= \frac{\ket{a}+\ket{b}}{\sqrt{2}} \frac{\bra{a}+\bra{b}}{\sqrt{2}} - \frac{\ket{a}-\ket{b}}{\sqrt{2}} \frac{\bra{a}-\bra{b}}{\sqrt{2}} \\
        &\triangleq \ket{a^\prime} \bra{a^\prime} - \ket{b^\prime} \bra{b^\prime},
    \end{aligned}
    \end{equation}
    where $\langle a^\prime\ket{b^\prime}=0$. Then, there exists a unitary matrix $B$ such that the following relation holds true,
    \begin{equation}
        B \ket{0}\ket{1}^{\otimes (n-1)} = \ket{a^\prime}, \quad B \ket{1}^{\otimes n} = \ket{b^\prime},
    \end{equation}
    and $S$ can be written as
    \begin{equation}\label{eqn:changebasis:S}
        S = B \left( Z \otimes \ket{1}\bra{1}^{\otimes (n-1)} \right) B^\dagger.
    \end{equation}
    The time evolution operator $\exp(iSt)$ can be implemented as
    \begin{equation}
        \exp(iSt) = B \cdot  CRZ^{1,\dots,n-1}_{n}(-2t) \cdot B^\dagger,
    \end{equation}
    where $CRZ^{1,\dots,n-1}_{n}(\theta)$ is the multi-controled RZ gate ($e^{-i\theta Z/2}$) acting on the $n$-th qubit controlled by $1,\dots,n-1$-th qubits.
\end{mylemma}

\begin{myremark}\label{remark:changebasis:lambda}
    By rotating $\ket{b}$ with the phase $e^{-i\lambda}$, similar results hold for $S = e^{i\lambda} \ket{a}\bra{b} + e^{-i\lambda} \ket{b}\bra{a}$. In particular, when $\ket{a}$ and $\ket{b}$ are basis states, we can explicitly construct the unitary matrix $B$ using H gates (Hadamard gates), Phase gates, Pauli gates, and CNOT gates, such as in Lemma \ref{lemma:evolution:bell}, Equation \eqref{eqn:adv:gate:U11}, Equation \eqref{eqn:adv:gate:U21}.
\end{myremark}

\begin{myremark}\label{remark:changebasis:norm}
    As reformulated in Equation \eqref{eqn:changebasis:S}, it follows that $\left\| S  \right\| = 1$ immediately. 
\end{myremark}

We introduce a useful  Lemma to facilitate the implementation of the differential operator's time evolution. The proof can be found in Appendix \ref{sec:appendix:circuit:lemma:bell}.
\begin{mylemma}{{\cite{sato2024hamiltonian}}}
\label{lemma:evolution:bell}
    The time evolution operator formulated as 
    \begin{equation}
        \exp \left( i\gamma\tau (e^{i\lambda} s_{j}^{-} + e^{-i\lambda} s_{j}^{+}) \right)
    \end{equation}
    can be implemented explicitly by
    \begin{equation}
        I^{\otimes(n_x-j)} \otimes W_{j}(\gamma\tau, \lambda),
    \end{equation}
    with 
    \begin{equation}
    \begin{aligned}
        W_{j}(\gamma\tau, \lambda) &:= B_{j}(\lambda) \text{CRZ}_{j}^{1,\dots,j-1}(-2\gamma\tau) B_{j}(\lambda)^{\dagger}, \\
        B_j(\lambda) &:= \left( \prod_{m=1}^{j-1} \text{CNOT}_{m}^{j} \right) P_{j}(-\lambda) H_{j},
    \end{aligned}
    \end{equation}
    where $H_{j}$ is the Hadamard gate acting on the $j$-th qubit, $P_{j}(\lambda)$ is the Phase gate acting on the $j$-th qubit as
    \begin{equation}
        P_{j}(\lambda) := \begin{bmatrix}
            1 & 0 \\ 0 & e^{i\lambda}
        \end{bmatrix},
    \end{equation}
    $\text{CNOT}_{m}^{j}$ is the CNOT gate acting on the $m$-th qubit controlled by the $j$-th qubit.
\end{mylemma}

\section{The heat equation}
\label{sec:heat}

We begin by considering a scalar field $u$ governed by the $d$-dimensional heat equation.
\begin{equation}\label{mod:heat}
    \left\{ \begin{aligned}
        \partial_t u(t,x) &= a \Delta u(t,x), \\
        u(0,x) &= f(x),
    \end{aligned} \right. \quad x\in \Omega \subset \mathbb{R}^d,
\end{equation}
with the Dirichlet boundary condition, where $a>0$ is a constant, the $d$-dimensional heat equation can be solved using the central difference method in classical simulation as follows:
\begin{equation}\label{mod:heat:num}
    d\mathbf{u}(t) / dt = a D_{D}^{\Delta} \mathbf{u}(t), \quad \mathbf{u}(0) = \mathbf{u}_0,
\end{equation}
where $\mathbf{u} := \left[u_{j_1,\dots, j_d}\right]_{1\leq j_i \leq N_x-1}$.

\subsection{Schr\"odingerisation}

The heat equation, although of dissipation type and not directly simulatable using Hamiltonian simulation, can be transformed into an equivalent Schr\"odinger-type equation using the Schr\"odingerization technique. This transformation enables the application of quantum algorithms designed for Schr\"odinger-type equations to solve the heat equation in one-higher dimension.

\subsubsection{Continuous formulation}
To be more precise, we can apply the warped phase transformation $v(t,x,p) = e^{-p} u(t,x)$, where $p \geq 0$, to the heat equation,
\begin{equation}\label{heat-sch}
    \partial_t v(t,x,p) = a \Delta v(t,x,p) = - a \partial_p \Delta v(t,x,p).
\end{equation}
By extending to $p<0$ with initial data $v(0,x,p) = e^{-|p|}u(0,x)$, and applying $\mathcal{F}{p}$, the Fourier transform over $p$, to Equation \eqref{heat-sch}, one obtains
\begin{equation}\label{mod:heat:eta}
    \partial_t \hat{v}(t,x,\eta) = i\eta a \Delta \hat{v}(t,x,\eta),
\end{equation}
where $\hat{v}(t,x,\eta) := \mathcal{F}_{p} v(t,x,p)$.
Since the operator $\eta a \Delta$ is self-adjoint under the Dirichlet BC, Equation \eqref{mod:heat:eta} is exactly the Schr\"odinger equation.

\subsubsection{Discretization}
The next step involves discretizing the problem. For instance, in the one-dimensional case, the variables $v$ and $\hat{v}$ are discretized as $\mathbf{v}(t) := [v_{j,k}(t)]_{j,k}$ and $\mathbf{\hat{v}}(t) := [\hat{v}_{j,k}(t)]_{j,k}$ with $v_{j,k}(t) := v(t, x_j, p_k)$, $\hat{v}_{j,k}(t) := v(t, x_j, \eta_k)$ respectively. The initial conditions can be represented as follows:
\begin{equation}\label{eqn:heat:initial}
\begin{aligned}
    \mathbf{v}(0) &:= \mathbf{u}(0) \otimes \mathbf{p} \\
    \mathbf{\hat{v}}(0) &:= \mathcal{F}_{p} \mathbf{v}(0) = \mathbf{u}(0) \otimes \mathbf{\eta} \\
    \mathbf{p} &:= \left[e^{-|p_0|}; e^{-|p_1|}; \dots; e^{- \left|p_{N_p-1}\right|} \right], \\
    \mathbf{\eta} &:= \left[\frac{2}{\eta_0^2 + 1}; \frac{2}{\eta_{1}^2+1}; \dots; \frac{2}{\eta_{N_p-1}^2 + 1}\right],
\end{aligned}
\end{equation}
where $-\pi R = p_0 < \cdots < p_{N_p} = \pi R$ with mesh size $\Delta p = 2\pi R/N_p$, $N_p=2^{n_p}$ and the Fourier variable $\eta_{k} = \left(k-\frac{N_p}{2}\right) / R, k=0, 1, \dots, N_p-1$. The Hamiltonian $\eta a \partial_{xx}$ is discretized using the central difference method as follows,
\begin{equation}
\begin{aligned}
    &\mathbf{H}_{\text{heat}} = a D_{D}^{\Delta} \otimes D_{\eta} \\
    =& \frac{a}{h^2} \left(S^{-} + S^{+} - 2 I^{\otimes n_x} \right) \otimes \diag\left(\eta_0,\dots,\eta_{N_p-1}\right) \\
    =& \sum_{k=0}^{N_p-1} \left(k-\frac{N_p}{2}\right) \mathbf{H}_{0} \otimes |k\rangle \langle k|,
\end{aligned}
\end{equation}
with
\begin{equation}\label{def:heat:H0gamma}
    \mathbf{H}_{0} := \gamma_{0} \left[ \sum_{j=1}^{n_x} ( s_{j}^{-} +  s_{j}^{+}) - 2 I^{\otimes n_x} \right], \  \gamma_{0} = \frac{a}{h^2 R}.
\end{equation}
In the $d$-dimensional case, one has
\begin{equation}\label{eqn:heat:Hd}
\begin{aligned}
    &\mathbf{H}_{\text{heat}} = a \sum_{\alpha=1}^{d} \left( D_{D}^{\Delta} \right)_{\alpha} \otimes D_{\eta} \\
    =& \sum_{k=0}^{N_p-1} \left(k-\frac{N_p}{2}\right) \sum_{\alpha=1}^{d} \left(\mathbf{H}_{0}\right)_{\alpha} \otimes |k\rangle \langle k| ,
\end{aligned}
\end{equation}
where $(\bullet)_{\alpha} := I^{\otimes (d-\alpha)n_x} \otimes \bullet \otimes I^{\otimes (\alpha-1) n_x}$.

\subsection{Quantum circuit}
\label{sec:heat:circuit}
We move forward to develop a detailed quantum circuit designed to solve the heat equation. To simplify the process and enhance the comprehension of the construction steps, we will adopt the notations listed below,
\begin{equation}
\begin{aligned}
    U_0(\tau) &:= \exp(i\mathbf{H}_{0}\tau), \\
    \tilde{U}_{0}(\tau) &:= \prod_{\alpha=1}^{d} \left(U_0(\tau)\right)_{\alpha}, \\
    U_{\text{heat}}(\tau) &:= \exp(i \mathbf{H}_{\text{heat}} \tau).
\end{aligned}
\end{equation}
Then we can obtain
\begin{equation}\label{eqn:heat:notation}
    U_{\text{heat}}(\tau) = \sum_{k=0}^{N_p-1} \tilde{U}_{0}^{k-N_p/2} (\tau) \otimes \ket{k}\bra{k}.
\end{equation}
The proof of this is provided in Appendix \ref{sec:appendix:circuit:heat:notation}.

We then employ $V_{0}, \tilde{V}_{0}, V_{\text{heat}}$ as the approximations of $U_{0}, \tilde{U}_{0}, U_{\text{heat}}$, respectively. Here, $V_0$ is derived by applying the first-order Lie-Trotter-Suzuki decomposition to $U_{0}$, specifically as
\begin{equation}\label{eqn:heat:V0:def}
\begin{aligned}
    & U_0(\tau) \\
    =& \exp \left( i \gamma_{0} \tau \sum_{j=1}^{n_x} ( s_{j}^{-} + s_{j}^{+}) -2i\gamma_{0} \tau I^{\otimes n_x}\right) \\
    \approx & \exp(-2i\gamma_{0} \tau) \prod_{j=1}^{n_x} \exp \left( i \gamma_{0} \tau ( s_{j}^{-} + s_{j}^{+})\right) \\
    \triangleq& V_0(\tau).
\end{aligned}
\end{equation}
Recalling Lemma \ref{lemma:evolution:bell} and setting $\gamma=\gamma_0, \lambda=0$, $V_0$ can be implemented by 
\begin{equation}\label{eqn:heat:V0:imple}
    V_0(\tau) = \text{Ph}(-2\gamma_{0}\tau) \prod_{j=1}^{n_x} I^{\otimes(n_x-j)} \otimes W_j(\gamma_{0}\tau, 0),
\end{equation}
where $\text{Ph}(\theta):=e^{i\theta} I^{\otimes n_x}$ is the global phase gate. Following that, $\tilde{V}_{0}$ and $V_{\text{heat}}$ can be expressed as
\begin{equation}\label{eqn:heat:gate:Vtau}
\begin{aligned}
    &\tilde{V}_{0}(\tau) := \prod_{\alpha=1}^{d} \left(V_0(\tau)\right)_{\alpha}, \\
    &V_{\text{heat}}(\tau) := \sum_{k=0}^{N_p-1} \tilde{V}_{0}^{k-N_p/2} (\tau) \otimes \ket{k}\bra{k}, \\
    =& \left( \tilde{V}^{-N_p/2}_{0}(\tau) \otimes I^{\otimes n_p} \right) \sum_{k=0}^{N_p-1} \tilde{V}_0^{k}(\tau) \otimes |k\rangle \langle k|.
\end{aligned}
\end{equation}
To effectively realize $V_{\text{heat}}$, we employ the subsequent method.
\begin{mylemma}
\label{lemma:H1:select}
    The unitary operator formulated as
    \begin{equation}
        Q = \sum_{k=0}^{N_p-1} Q_0^k \otimes \ket{k}\bra{k},    
    \end{equation}
    can be realized by performing the controlled operator $c-Q_0^{2^m}$ on the first register, controlled by the $m$-th qubit of the second register.
\end{mylemma}
\begin{myremark}
    In the most favorable scenario, when $Q_0^{k}(\tau)$ can be implemented with cost independent on $k$—such as $Q_0^{k}(\tau) = Q_0(k\tau)$—this strategy is significantly more efficient compared to implementing $Q(\tau)$ directly, as the efficiency is exponential. Nevertheless, in this study, we find that $\tilde{V}_{0}^k(\tau) \neq \tilde{V}_{0}(k\tau)$, resulting in $O(N_p)$ operations which offer a quadratic speed-up relative to directly executing $V_{\text{heat}}$.
\end{myremark}
The proof for Lemma \ref{lemma:H1:select} can be found in Appendix \ref{sec:appendix:circuit:lemma:H1:select}. Detailed quantum circuit designs for the unitary matrix $W_{j}(\gamma\tau, \lambda)$, the operators $V_{0}(\tau)$, $\tilde{V}_{0}(\tau)$, and $V_{\text{heat}}(\tau)$ are depicted in Figure \ref{fig:heat:circuit:Wj}, Figure \ref{fig:heat:circuit:V0}, Figure \ref{fig:heat:circuit:tildeV0}, and Figure \ref{fig:heat:circuit:V}, respectively.

Given $T=r\tau$, the operator $V_{\text{heat}}(\tau)$ is iterated $r$ times to give the state $\ket{\hat{\mathbf{v}}_{D}(T)} := V_{\text{heat}}^{r}(\tau) \ket{\hat{\mathbf{v}}(0)}$. An inverse quantum Fourier transform can then be conducted to retrieve $\ket{\mathbf{v}_{D}(T)}$. Subsequently, a measurement 
\begin{equation}
    M_{\geq 0} = \sum_{p_k \geq 0} M_{k} := \sum_{p_k \geq 0} I^{\otimes n_x d} \otimes \ket{k} \bra{k},
\end{equation}
is taken to select only the $\ket{k}, p_k \geq 0$ part of the state $\ket{\mathbf{v}_{D}(T)}$. Since $v(t,x,p) = e^{-p} u(t,x)$ for $p \geq 0$, it follows that
\begin{equation}
\begin{aligned}
    & \bra{\mathbf{v}_{D}(T)} M_{k} \ket{\mathbf{v}_{D}(T)} \approx \frac{e^{-2p_k} \left\| \mathbf{u}(T) \right\|^2}{\left\| \mathbf{v}_{D}(T) \right\|^2} \\
    =& \frac{e^{-2p_k} \left\| \mathbf{u}(T) \right\|^2}{\left\| \mathbf{v}(0) \right\|^2} = \frac{e^{-2p_k} \left\| \mathbf{u}(T) \right\|^2}{\left\| \mathbf{p} \right\|^2 \left\| \mathbf{u}(0) \right\|^2 }.
\end{aligned}
\end{equation}
After the measurement, the output is an approximation of $\ket{\mathbf{u}(T)} \ket{k}$ with probability $e^{-2p_k} \left\| \mathbf{u}(T) \right\|^2 / \left\| \mathbf{p} \right\|^2 \left\| \mathbf{u}(0) \right\|^2$. The full circuit to implement the Schr\"odingerisation method is shown in Figure \ref{fig:circuit:schro}, where $\mathcal{QFT}$ ($\mathcal{IQFT}$) denotes the (inverse) quantum Fourier transform.

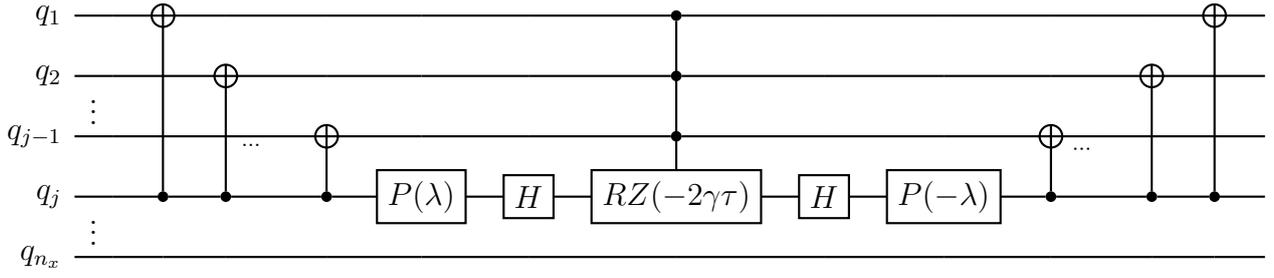
\begin{figure*}[htbp]
    \centering
    \begin{quantikz}[row sep={0.8cm,between origins}]
        \lstick{$q_1$} & \qw &  \targ{} & \qw & \qw & \qw & \qw & \qw & \ctrl{3} & \qw & \qw & \qw & \qw & \qw &\targ{} &\qw \\
        \lstick{$q_2$} & \qw{\vdots}  & \qw & \targ{} & \qw & \qw & \qw & \qw & \control{} & \qw & \qw & \qw & \qw & \targ{} &\qw & \qw \\
        \lstick{$q_{j-1}$} & \qw &  \qw & \qw &  \qw{\ldots} & \targ{} & \qw & \qw & \control{} & \qw & \qw & \targ{} & \qw{\cdots} & \qw & \qw & \qw \\
        \lstick{$q_{j}$} & \qw{\vdots}  & \ctrl{-3} & \ctrl{-2} & \qw & \ctrl{-1} & \gate{P(\lambda)} & \gate{H} & \gate{RZ(-2\gamma\tau)} & \gate{H} & \gate{P(-\lambda)} & \ctrl{-1} & \qw & \ctrl{-2} & \ctrl{-3} & \qw\\
        \lstick{$q_{n_x}$} & \qw & \qw &  \qw & \qw & \qw & \qw & \qw & \qw & \qw& \qw & \qw & \qw & \qw & \qw & \qw
    \end{quantikz}
    \caption{Quantum circuit for $W_j(\gamma\tau, \lambda)$.}
    \label{fig:heat:circuit:Wj}
\end{figure*}

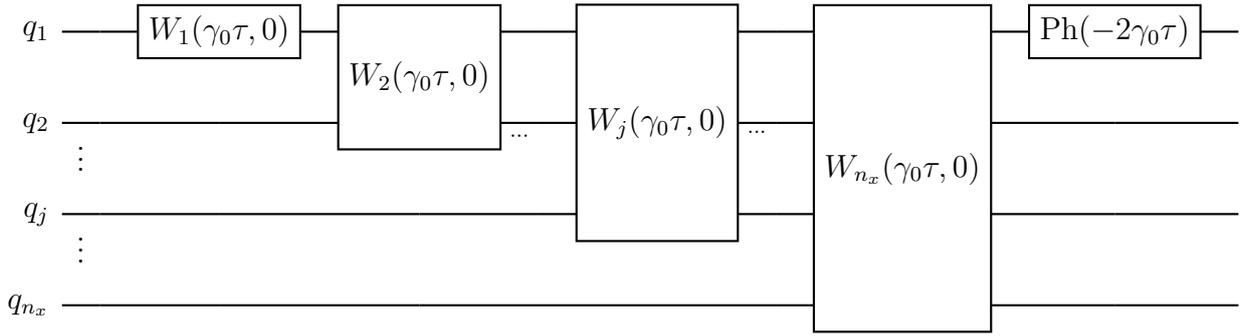
\begin{figure*}[htbp]
    \centering
    \begin{quantikz}
        \lstick{$q_1$} & \qw & \gate{W_{1}(\gamma_{0}\tau,0)} &  \gate[2]{W_{2}(\gamma_{0}\tau,0)} & \qw & \gate[3]{W_{j}(\gamma_{0}\tau,0)} & \qw & \gate[4]{W_{n_x}(\gamma_{0}\tau,0)} & \gate{\text{Ph}(-2\gamma_{0}\tau)} & \qw \\
        \lstick{$q_2$} & \qw{\vdots}  & \qw &  & \qw{\cdots} &  & \qw{\ldots} &  & \qw & \qw  \\
        \lstick{$q_{j}$} & \qw{\vdots}  & \qw & \qw & \qw &  & \qw &  & \qw & \qw \\
        \lstick{$q_{n_x}$} & \qw & \qw & \qw & \qw & \qw & \qw &  & \qw & \qw 
    \end{quantikz}
    \caption{Quantum circuit for $V_0(\tau)$.}
    \label{fig:heat:circuit:V0}
\end{figure*}

\begin{figure*}[htbp]
    \centering
    \begin{quantikz}[row sep={0.8cm,between origins}]
        \lstick{$q^{1}$} & \qwbundle{n_x} & \gate[4]{\tilde{V}_{0}(\tau)} & \qw \\
        \lstick{$q^{2}$} & \qwbundle{n_x}{\vdots} & & \qw \\
        \lstick{$q^{j}$} & \qwbundle{n_x}{\vdots} & & \qw \\
        \lstick{$q^{d}$} & \qwbundle{n_x} & & \qw \\
    \end{quantikz} 
    :=
    \begin{quantikz}
        \lstick{$q^{1}$} & \qwbundle{n_x} & \gate{V_{0}(\tau)} &  \qw & \qw & \qw & \qw & \qw & \qw \\
        \lstick{$q^{2}$} & \qwbundle{n_x}{\vdots}  & \qw & \gate{V_{0}(\tau)} & \qw{\ldots} & \qw  & \qw & \qw  & \qw \\
        \lstick{$q^{j}$} & \qwbundle{n_x}{\vdots}  & \qw & \qw & \qw & \gate{V_{0}(\tau)} & \qw{\ldots} & \qw & \qw \\
        \lstick{$q^{d}$} & \qwbundle{n_x} & \qw & \qw & \qw & \qw & \qw & \gate{V_{0}(\tau)} & \qw 
    \end{quantikz}
    \caption{Quantum circuit for $\tilde{V}_0(\tau)$.}
    \label{fig:heat:circuit:tildeV0}
\end{figure*}
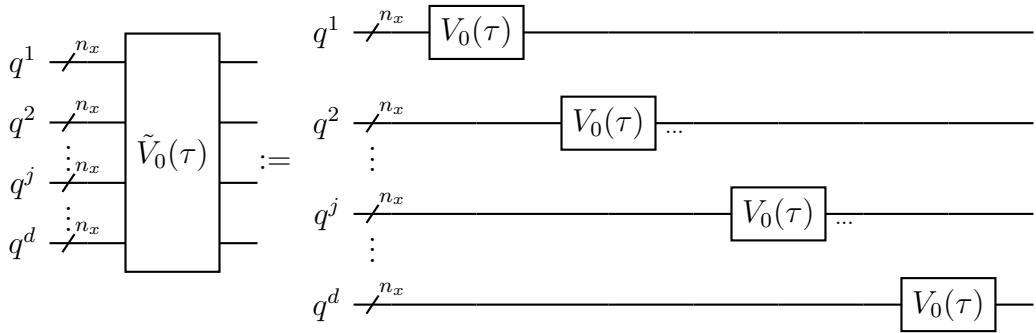

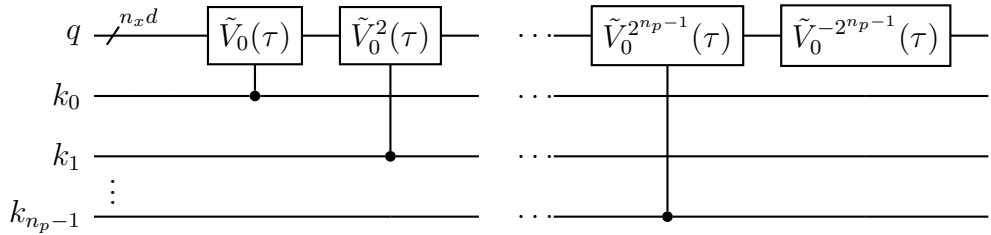
\begin{figure*}[htbp]
    \centering
    \begin{quantikz}[row sep={0.8cm,between origins}]
        \lstick{$q$} & \qwbundle{n_x d} & \qw & \gate{\tilde{V}_{0}(\tau)} & \gate{\tilde{V}_{0}^{2}(\tau)} & \qw & \ldots & \gate{\tilde{V}_{0}^{2^{n_p-1}}(\tau)} & \gate{\tilde{V}_{0}^{-2^{n_p-1}}(\tau)} & \qw \\
        \lstick{$k_0$} & \qw & \qw & \ctrl{-1} & \qw & \qw & \ldots & \qw & \qw & \qw \\
        \lstick{$k_1$} & \qw{\vdots} & \qw & \qw & \ctrl{-2} & \qw & \ldots & \qw & \qw & \qw \\
        \lstick{$k_{n_p-1}$} & \qw & \qw & \qw & \qw & \qw & \ldots & \ctrl{-3} & \qw & \qw \\
    \end{quantikz} 
    \caption{Quantum circuit for $V_{\text{heat}}(\tau)$.}
    \label{fig:heat:circuit:V}
\end{figure*}

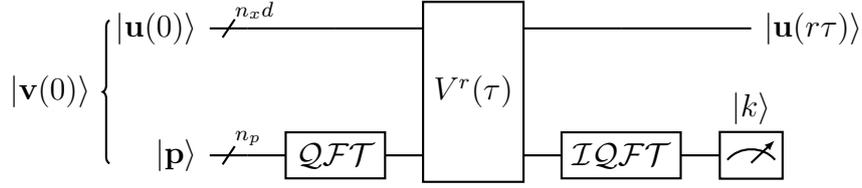
\begin{figure*}[htbp]
    \centering
    \begin{quantikz}
        \lstick[2]{$\ket{\mathbf{v}(0)}$}\,\,\, &  & \lstick { $| \mathbf{u}(0) \rangle$ } & \qwbundle{n_x d} & \qw & \gate[2]{V^r(\tau)} & \qw & \qw\rstick{$\ket{\mathbf{u}(r\tau)} $}\\
        & & \lstick{$\ket{\mathbf{p}}$} & \qwbundle{n_p} & \gate{\mathcal{QFT}} & & \gate{\mathcal{IQFT}} & \meter{$\ket{k}$} \\
    \end{quantikz} 
    \caption{Quantum circuit for the Schr\"odingerisation method, where the measurement requires only projection onto $\ket{k}, p_k > 0$, and $\mathcal{QFT}$ ($\mathcal{IQFT}$) denotes the (inverse) quantum Fourier transform.}
    \label{fig:circuit:schro}
\end{figure*}

\section{The advection equation}
\label{sec:advection}

For the second illustrative example, let us consider the scalar field $u$ governed by the $d$-dimensional advection equation,
\begin{equation}\label{mod:adv}
    \partial_t u(t,x) = \mathbf{a} \cdot \nabla u(t,x) = \sum_{\alpha=1}^{d} a_{\alpha}\partial_{x_{\alpha}} u(t,x),
\end{equation}
with periodic boundary conditions, where $\mathbf{a}$ is a constant vector field. For simplicity, we assume that $a_{\alpha}>0$ for $\alpha=1,\dots,d_0$, and $a_{\alpha}<0$ for $\alpha=d_0+1,\dots,d$. In classical simulation, Equation \eqref{mod:adv} can be solved using the upwind method as follows,
\begin{equation}\label{mod:adv:num}
\begin{aligned}
    & d\mathbf{u}(t) / dt \\
    =& \left( \sum_{\alpha=1}^{d_0} a_{\alpha}  \left( D_{P}^{+} \right)_{\alpha} + \sum_{\alpha=d_0+1}^{d} a_{\alpha} \left(D_{P}^{-}\right)_{\alpha} \right) \mathbf{u}(t) \\
    \triangleq & \mathbf{A} \mathbf{u}(t).
\end{aligned}
\end{equation}

In \cite{sato2024hamiltonian}, the central finite difference method was employed to solve this equation while aiming to preserve its Schr\"odinger structure. However, it is important to note that the central finite difference method introduces significant dispersion errors, leading to spurious numerical oscillations across discontinuities, which was shown in the numerical examples in \cite{sato2024hamiltonian}, thus is not the preferred approach for the advection equation when the solution becomes discontinuous. The center difference is natural for quantum simulation since its discrete spectral is purely imaginary, corresponding to unitary dynamics.  The upwind scheme, on the other hand, is dissipative, thus one cannot directly solve it via quantum simulation. Therefore we will utilize the Schr\"odingerization technique. 

\subsection{Schr\"odingerisation}

In this case, $-i\mathbf{A}$ is not Hermitian, so we decompose $\mathbf{A}$ into $\mathbf{A} = \mathbf{A}_1 + i \mathbf{A}_{2}$ with
\begin{equation}
\begin{aligned}
    \mathbf{A}_1 &= \frac{\mathbf{A} + \mathbf{A}^\dagger}{2}, \quad \mathbf{A}_{2} = \frac{\mathbf{A} - \mathbf{A}^\dagger}{2i}, \\
    \mathbf{A}_1 & =\mathbf{A}_1^\dagger, \quad \mathbf{A}_{2}=\mathbf{A}_{2}^\dagger.
\end{aligned}
\end{equation}

\subsubsection{Continuous formulation}
Applying the Schr\"odingerisation method, the warped phase transformation $\mathbf{v}(t,p) = e^{-p} \mathbf{u}(t)$, $p \geq 0$, satisfies
\begin{equation}
\begin{aligned}
\partial_t \mathbf{v}(t,p) &= \mathbf{A}_1 \mathbf{v}(t,p) + i \mathbf{A}_2\mathbf{v}(t,p) \\
&= - \mathbf{A}_1 \partial_p \mathbf{v}(t,p) + i \mathbf{A}_2\mathbf{v}(t,p),
\end{aligned}
\end{equation}
and this equation is extended to $p<0$ with initial data $\mathbf{v}(0,p) = e^{-|p|}\mathbf{u}(0)$. After performing the Fourier transform with respect to $p$, one obtains
\begin{equation}\label{mod:adv:eta}
    \partial_t \hat{\mathbf{v}}(t,\eta) = i (\eta \mathbf{A}_{1} + \mathbf{A}_{2}) \hat{\mathbf{v}}(t,\eta) \triangleq i \mathbf{H} \hat{\mathbf{v}}(t,\eta),
\end{equation}
with $\mathbf{H}$ being Hermitian.

\subsubsection{Discretization}

Similarly to the heat equation, we discretize $p$ as $-\pi R = p_0 < \cdots < p_{N_p} = \pi R$ with a mesh size of $\Delta p = 2\pi R/N_p$. We also define $\eta_{k} = \left(k-\frac{N_p}{2}\right) / R$ for $k=0, 1, \dots, N_p-1$. The initial conditions can be represented as,
\begin{equation}
\begin{aligned}
    \mathbf{v}(0) &:= \mathbf{u}(0) \otimes \mathbf{p}, \\
    \mathbf{\hat{v}}(0) &:= \mathcal{F}_{p} \mathbf{v}(0) = \mathbf{u}(0) \otimes \mathbf{\eta},
\end{aligned}
\end{equation}
with $\mathbf{u}(0):=[u_0(0);u_1(0);\dots;u_{N_x -1}(0)]$ and $\mathbf{p}$, $\mathbf{\eta}$ defined as in \eqref{eqn:heat:initial}. Noting that
\begin{equation}
\begin{aligned}
    \mathbf{A}_{1} &= \sum_{\alpha=1}^{d} \frac{\left|a_{\alpha}\right| h}{2} \left( D_{P}^{\Delta} \right)_{\alpha}, \\
    \mathbf{A}_{2} &= \sum_{\alpha=1}^{d} \frac{a_{\alpha}}{i} \left( D_{P}^{\pm} \right)_{\alpha}.
\end{aligned}
\end{equation}

The Hamiltonian $\mathbf{H}_{\text{adv}}$ is discretized as
\begin{equation}\label{eqn:adv:Hd}
\begin{aligned}
    &\mathbf{H}_{\text{adv}} = \mathbf{A}_{1} \otimes D_{\eta} + \mathbf{A}_{2} \otimes I^{\otimes n_p} \\
    =& \sum_{k=0}^{N_p-1} \left( k - \frac{N_p}{2}\right) \sum_{\alpha=1}^{d} |a_{\alpha}| \left( \mathbf{H}_{1} \right)_{\alpha} \otimes \ket{k}\bra{k} \\
    &+ \sum_{\alpha=1}^{d} a_{\alpha} \left( \mathbf{H}_{2} \right)_{\alpha}  \otimes I^{\otimes n_p}.
\end{aligned}
\end{equation}
with
\begin{equation}\label{def:adv:H1gamma}
\begin{aligned}
    \mathbf{H}_{1} &= \mathbf{H}_{1}^{(1)} + \mathbf{H}_{1}^{(2)}, \quad \gamma_{1} = \frac{1}{2hR}, \\ 
    \mathbf{H}_{1}^{(1)} &:= \gamma_{1} ( \sigma_{01}^{\otimes n_x} +  \sigma_{10}^{\otimes n_x}), \\
    \mathbf{H}_{1}^{(2)} &:= \gamma_{1} \sum_{j=1}^{n_x} ( s_{j}^{-} +  s_{j}^{+}) -2 \gamma_{1} I^{\otimes n_x},
\end{aligned}
\end{equation}

\begin{equation}\label{def:adv:H2gamma}
\begin{aligned}
    \mathbf{H}_{2} &= \mathbf{H}_{2}^{(1)} + \mathbf{H}_{2}^{(2)}, \quad \gamma_{2} = \frac{1}{2h}, \\
    \mathbf{H}_{2}^{(1)} &:= -i\gamma_{2} ( -\sigma_{01}^{\otimes n_x} +  \sigma_{10}^{\otimes n_x}), \\
    \mathbf{H}_{2}^{(2)} &:= -i \gamma_{2} \sum_{j=1}^{n_x} ( s_{j}^{-} - s_{j}^{+}).
\end{aligned}
\end{equation}

\subsection{Quantum circuit}
\label{sec:adv:circuit}
Similar to Section \ref{sec:heat:circuit}, we introduce the following notations to clarify the construction procedure,
\begin{equation}
\begin{aligned}
    &U_1(\tau) := \exp(i\mathbf{H}_{1}\tau), \  \tilde{U}_{1}(\tau) := \prod_{\alpha=1}^{d} \left(U_1(|a_\alpha|\tau)\right)_{\alpha} \\
    &U_2(\tau) := \exp(i\mathbf{H}_{2}\tau), \  \tilde{U}_{2}(\tau) := \prod_{\alpha=1}^{d} \left(U_2(a_\alpha\tau)\right)_{\alpha} \\
    &U_{\text{adv}}(\tau) := \exp(i \mathbf{H}_{\text{adv}} \tau).
\end{aligned}
\end{equation}
We then obtain the following result, with the proof provided in Appendix \ref{sec:appendix:circuit:adv:notation}.
\begin{equation}\label{eqn:adv:notation}
\begin{aligned}
    & U_{\text{adv}}(\tau) \approx U_{*}(\tau) :=\\
    & \left( \tilde{U}_{2}(\tau) \otimes I^{\otimes n_p} \right) \sum_{k=0}^{N_p-1} \tilde{U}_{1}^{k-N_p/2}(\tau) \otimes |k\rangle \langle k|,
\end{aligned}
\end{equation}
the proof of which is presented in Appendix \ref{sec:appendix:circuit:adv:notation}. Additionally, We denote $U_{1}^{(1)}, U_{1}^{(2)}, U_{1}^{(1)}, U_{1}^{(2)}$ as the Hamiltonian simulation operators corresponding to $\mathbf{H}_{1}^{(1)}, \mathbf{H}_{1}^{(2)}, \mathbf{H}_{2}^{(1)}, \mathbf{H}_{2}^{(2)}$, respectively. The approximations of these operators are denoted as the corresponding `$V$' operators.

The first-order Lie-Trotter-Suzuki decomposition of $U_{1}$ and $U_{2}$ gives
\begin{equation}\label{eqn:adv:unitary:U1U2}
\begin{aligned}
    U_1(\tau) &= \exp(i(\mathbf{H}_{1}^{(1)} + \mathbf{H}_{1}^{(2)})t) \\
    &\approx U_{1}^{(1)}(\tau) U_{1}^{(2)}(\tau), \\
    U_2(\tau) &= \exp(i(\mathbf{H}_{2}^{(1)} + \mathbf{H}_{2}^{(2)})t) \\
    &\approx U_{2}^{(1)}(\tau) U_{2}^{(2)}(\tau).
\end{aligned}
\end{equation}
Similar to Equation \eqref{eqn:heat:V0:def} and Equation \eqref{eqn:heat:V0:imple}, $U_1^{(2)}, U_2^{(2)}$ are approximated by 
\begin{equation}\label{eqn:adv:gate:V12}
\begin{aligned}
    &U_{1}^{(2)}(\tau) \approx V_1^{(2)}(\tau) := \\
    &\text{Ph}(-2\gamma_{1}\tau) \prod_{j=1}^{n_x} I^{\otimes(n_x-j)} \otimes W_j(\gamma_{1}\tau, 0), \\
    &U_{2}^{(2)}(\tau) \approx V_{2}^{(2)}(\tau) := \\
    &\prod_{j=1}^{n_x} I^{\otimes(n_x-j)} \otimes W_j\left(\gamma_{2}\tau, -\frac{\pi}{2}\right),
\end{aligned}
\end{equation}
by taking $\gamma = \gamma_{1}$, $\lambda=0$ and $\gamma=\gamma_{2}$, $\lambda=-\pi/2$ in Lemma \ref{lemma:evolution:bell}.

The quantum circuits for $U_{1}^{(1)}(\tau)$ and $U_{2}^{(1)}(\tau)$ can be explicitly constructed following Lemma \ref{lemma:changebasis} as
\begin{equation}\label{eqn:adv:gate:U11}
\begin{aligned}
    U_{1}^{(1)}(\tau) &= \exp \left( i\gamma_{1}\tau \left(\sigma_{01}^{\otimes n_x} + \sigma_{10}^{\otimes n_x}\right) \right) \\
    &= B^{(1)} \cdot \text{CRZ}_{n_x}^{1,\dots,n_x-1}(-2\gamma_{1}\tau) \cdot \left(B^{(1)}\right)^\dagger, \\
    B^{(1)}:&=\prod_{m=1}^{n_x-1}\text{CNOT}_{m}^{n_x} H_{n_x}\prod_{m=1}^{n_x-1} X_{m} \\
    &= B_{n_x}(0) \prod_{m=1}^{n_x-1} X_{m}.
\end{aligned}
\end{equation}
\begin{equation}\label{eqn:adv:gate:U21}
\begin{aligned}
    U_{2}^{(1)}(\tau) &= \exp \left( i\gamma_{2}\tau \left(i\sigma_{01}^{\otimes n_x} - i\sigma_{10}^{\otimes n_x}\right) \right) \\
    &= B^{(2)} \cdot \text{CRZ}_{n_x}^{1,\dots,n_x-1}(-2\gamma_2\tau) \cdot \left(B^{(2)}\right)^\dagger, \\
    B^{(2)}:&=\prod_{m=1}^{n_x-1}\text{CNOT}_{m}^{n_x} P_{n_x}\left(-\frac{\pi}{2} \right) H_{n_x} \prod_{m=1}^{n_x-1} X_{m} \\
    &= B_{n_x}\left( \frac{\pi}{2} \right) \prod_{m=1}^{n_x-1} X_{m}.
\end{aligned}
\end{equation}

Given $U_{1}^{(1)}$, $U_{2}^{(1)}$, $V_{1}^{(2)}$, $V_{2}^{(2)}$ defined in Equation \eqref{eqn:adv:gate:U11}, Equation \eqref{eqn:adv:gate:U21} and Equation \eqref{eqn:adv:gate:V12}, $U_{1}$, $U_{2}$, $\tilde{U}_{1}$ and $\tilde{U}_{2}$ are approximated by
\begin{equation}\label{eqn:adv:gate:V1V2}
\begin{aligned}
    U_{1}(\tau) &\approx V_{1}(\tau) := U_{1}^{(1)}(\tau) V_{1}^{(2)}(\tau), \\
    U_{2}(\tau) &\approx V_{2}(\tau) := U_{2}^{(1)}(\tau) V_{2}^{(2)}(\tau), \\
    \tilde{U}_{1}(\tau) &\approx \tilde{V}_{1}(\tau) := \prod_{\alpha=1}^{d} \left(V_1(|a_\alpha|\tau)\right)_{\alpha}, \\
    \tilde{U}_{2}(\tau) &\approx \tilde{V}_{2}(\tau) := \prod_{\alpha=1}^{d} \left(V_2(a_\alpha\tau)\right)_{\alpha}.
\end{aligned}
\end{equation}
Ultimately, $U_{\text{adv}}(\tau)$ is approximated by
\begin{equation}\label{eqn:adv:gate:Vtau}
\begin{aligned}
    &U_{\text{adv}}(\tau) \approx V_{\text{adv}}(\tau) := \\
    &\left( \tilde{V}_{2}(\tau) \otimes I^{\otimes n_p} \right) \sum_{k=0}^{N_p-1} \tilde{V}_{1}^{k-N_p/2}(\tau) \otimes |k\rangle \langle k|,
\end{aligned}
\end{equation}
The efficient implementation technique described in Lemma \ref{lemma:H1:select} is also applicable here. The detailed quantum circuits for the approximating unitary matrices $U_{1}^{(1)}(\tau)$, $U_{2}^{(1)}(\tau)$, $V_{1}(\tau)$, $V_{2}(\tau)$, $\tilde{V}_{1}(\tau)$, $\tilde{V}_{2}(\tau)$ and $V_{\text{adv}}(\tau)$ are shown in Figure \ref{fig:adv:circuit:U11}-\ref{fig:adv:circuit:V}. Given $V_{\text{adv}}(\tau)$, the circuit for implementing the Schr\"odingerisation method remains the same as in Figure \ref{fig:circuit:schro}.

\begin{figure*}[htbp]
    \centering
    \begin{quantikz}[row sep={0.8cm,between origins}]
        \lstick{$q_1$} & \qw &  \targ{} & \qw & \qw & \qw & \gate{X} & \ctrl{3} & \gate{X} & \qw & \qw & \qw &\targ{} &\qw \\
        \lstick{$q_2$} & \qw{\vdots}  & \qw & \targ{} & \qw & \qw & \gate{X} & \control{} & \gate{X} & \qw & \qw & \targ{} &\qw & \qw \\
        \lstick{$q_{n_x-1}$} & \qw &  \qw & \qw &  \qw{\ldots} & \targ{} & \gate{X} & \control{} & \gate{X} & \targ{} & \qw{\cdots} & \qw & \qw & \qw \\
        \lstick{$q_{n_x}$} & \qw  & \ctrl{-3} & \ctrl{-2} & \qw & \ctrl{-1} & \gate{H} & \gate{RZ(-2\gamma_{1}\tau)} & \gate{H} & \ctrl{-1} & \qw & \ctrl{-2} & \ctrl{-3} & \qw
    \end{quantikz}
    \caption{Quantum circuit for $U_{1}^{(1)}(\tau)$.}
    \label{fig:adv:circuit:U11}
\end{figure*}
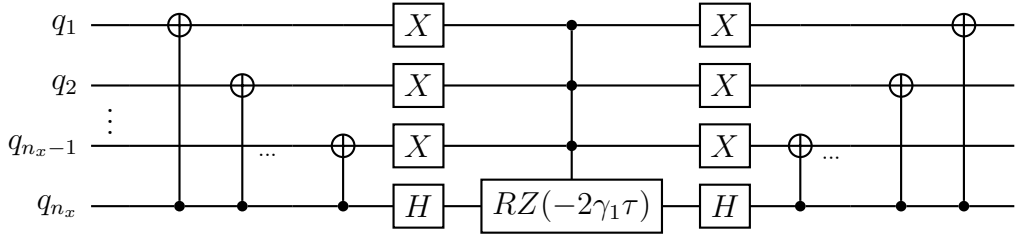

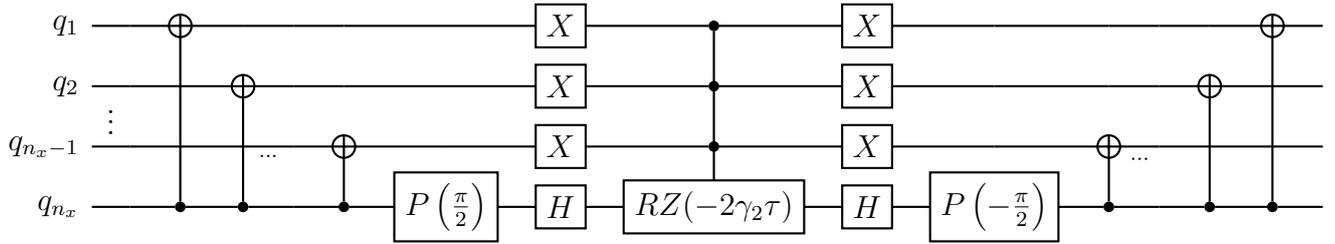
\begin{figure*}[htbp]
    \centering
    \begin{quantikz}[row sep={0.8cm,between origins}]
        \lstick{$q_1$} & \qw &  \targ{} & \qw & \qw & \qw & \qw & \gate{X} & \ctrl{3} & \gate{X} & \qw & \qw & \qw & \qw &\targ{} &\qw \\
        \lstick{$q_2$} & \qw{\vdots}  & \qw & \targ{} & \qw & \qw & \qw & \gate{X} & \control{} & \gate{X} & \qw & \qw & \qw & \targ{} &\qw & \qw \\
        \lstick{$q_{n_x-1}$} & \qw &  \qw & \qw &  \qw{\ldots} & \targ{} & \qw & \gate{X} & \control{} & \gate{X} & \qw & \targ{} & \qw{\cdots} & \qw & \qw & \qw \\
        \lstick{$q_{n_x}$} & \qw  & \ctrl{-3} & \ctrl{-2} & \qw & \ctrl{-1} & \gate{P\left( \frac{\pi}{2} \right)} & \gate{H} & \gate{RZ(-2\gamma_{2}\tau)} & \gate{H} & \gate{P\left( -\frac{\pi}{2} \right)} & \ctrl{-1} & \qw & \ctrl{-2} & \ctrl{-3} & \qw
    \end{quantikz}
    \caption{Quantum circuit for $U_{2}^{(1)}(\tau)$.}
    \label{fig:adv:circuit:U21}
\end{figure*}

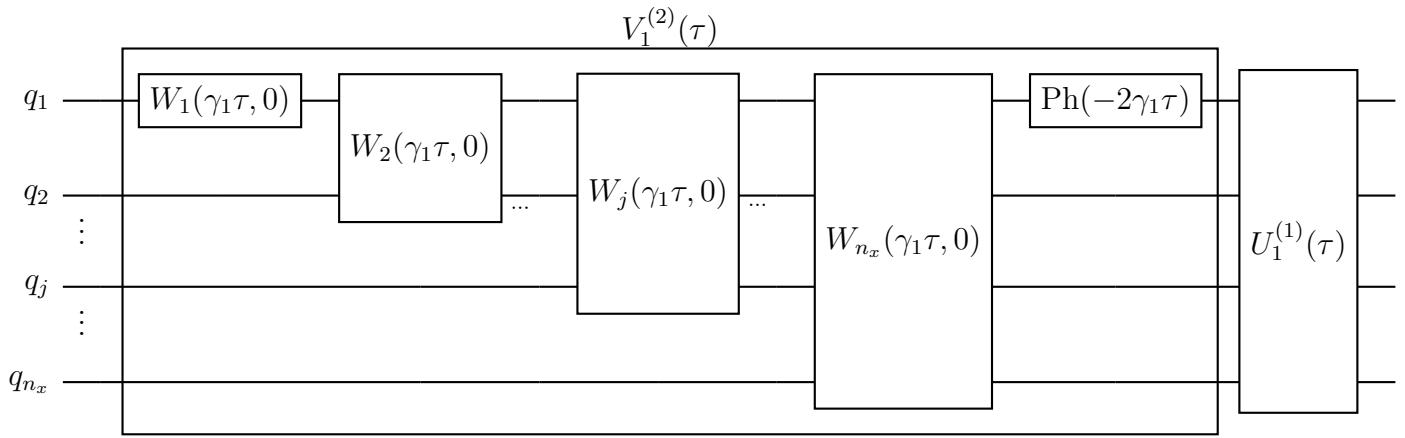
\begin{figure*}[htbp]
    \centering
    \begin{quantikz}
        \lstick{$q_1$} & \qw & \gate{W_{1}(\gamma_{1}\tau,0)}\gategroup[4,steps=7,style={inner xsep=2pt}, background]{$V_{1}^{(2)}(\tau)$} &  \gate[2]{W_{2}(\gamma_{1}\tau,0)} & \qw & \gate[3]{W_{j}(\gamma_{1}\tau,0)} & \qw & \gate[4]{W_{n_x}(\gamma_{1}\tau,0)} & \gate{\text{Ph}(-2\gamma_{1}\tau)} & \gate[4]{U_{1}^{(1)}(\tau)} & \qw \\
        \lstick{$q_2$} & \qw{\vdots}  & \qw &  & \qw{\cdots} &  & \qw{\ldots} &  & \qw & & \qw \\
        \lstick{$q_{j}$} & \qw{\vdots} & \qw & \qw & \qw &  & \qw &  & \qw &  & \qw \\
        \lstick{$q_{n_x}$}  & \qw & \qw & \qw & \qw & \qw & \qw &  & \qw &  & \qw 
    \end{quantikz}
    \caption{Quantum circuit for $V_1(\tau)$.}
    \label{fig:adv:circuit:V1}
\end{figure*}

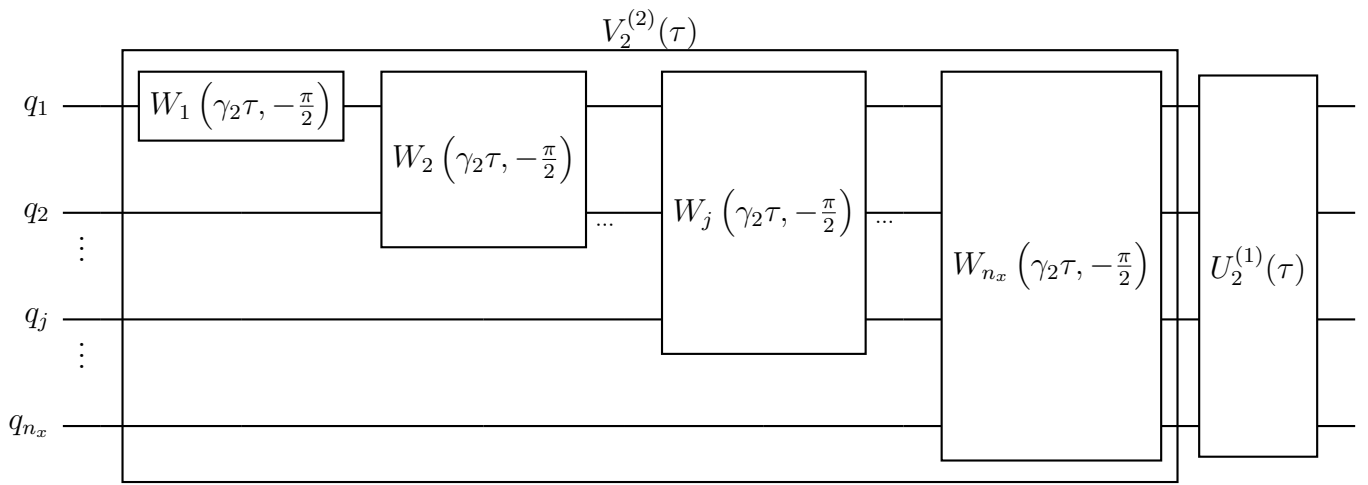
\begin{figure*}[htbp]
    \centering
    \begin{quantikz}
        \lstick{$q_1$} & \qw & \gate{W_{1}\left(\gamma_{2}\tau,-\frac{\pi}{2}\right)}\gategroup[4,steps=6,style={inner xsep=2pt}, background]{$V_{2}^{(2)}(\tau)$} &  \gate[2]{W_{2}\left(\gamma_{2}\tau,-\frac{\pi}{2}\right)} & \qw & \gate[3]{W_{j}\left(\gamma_{2}\tau,-\frac{\pi}{2}\right)} & \qw & \gate[4]{W_{n_x}\left(\gamma_{2}\tau,-\frac{\pi}{2}\right)} & \gate[4]{U_{2}^{(1)}(\tau)} & \qw \\
        \lstick{$q_2$} & \qw{\vdots}  & \qw &  & \qw{\cdots} &  & \qw{\ldots} &  &  & \qw  \\
        \lstick{$q_{j}$} & \qw{\vdots} & \qw & \qw & \qw &  & \qw &  &  & \qw \\
        \lstick{$q_{n_x}$}  & \qw & \qw & \qw & \qw & \qw & \qw &  &  & \qw 
    \end{quantikz}
    \caption{Quantum circuit for $V_2(\tau)$.}
    \label{fig:adv:circuit:V2}
\end{figure*}

\begin{figure*}[htbp]
    \centering
    \begin{quantikz}[row sep={0.8cm,between origins}]
        \lstick{$q^{1}$} & \qwbundle{n_x} & \gate[4]{\tilde{V}_{1}(\tau)} & \qw \\
        \lstick{$q^{2}$} & \qwbundle{n_x}{\vdots} & & \qw \\
        \lstick{$q^{j}$} & \qwbundle{n_x}{\vdots} & & \qw \\
        \lstick{$q^{d}$} & \qwbundle{n_x} & & \qw \\
    \end{quantikz} 
    :=
    \begin{quantikz}
        \lstick{$q^{1}$} & \qwbundle{n_x} & \gate{V_{1}(|a_{1}|\tau)} &  \qw & \qw & \qw & \qw & \qw & \qw \\
        \lstick{$q^{2}$} & \qwbundle{n_x}{\vdots}  & \qw & \gate{V_{1}(|a_{2}|\tau)} & \qw{\ldots} & \qw  & \qw & \qw  & \qw \\
        \lstick{$q^{j}$} & \qwbundle{n_x}{\vdots}  & \qw & \qw & \qw & \gate{V_{1}(|a_{j}|\tau)} & \qw{\ldots} & \qw & \qw \\
        \lstick{$q^{d}$} & \qwbundle{n_x} & \qw & \qw & \qw & \qw & \qw & \gate{V_{1}(|a_{n_x}|\tau)} & \qw 
    \end{quantikz}
    \caption{Quantum circuit for $\tilde{V}_1(\tau)$.}
    \label{fig:adv:circuit:tildeV1}
\end{figure*}
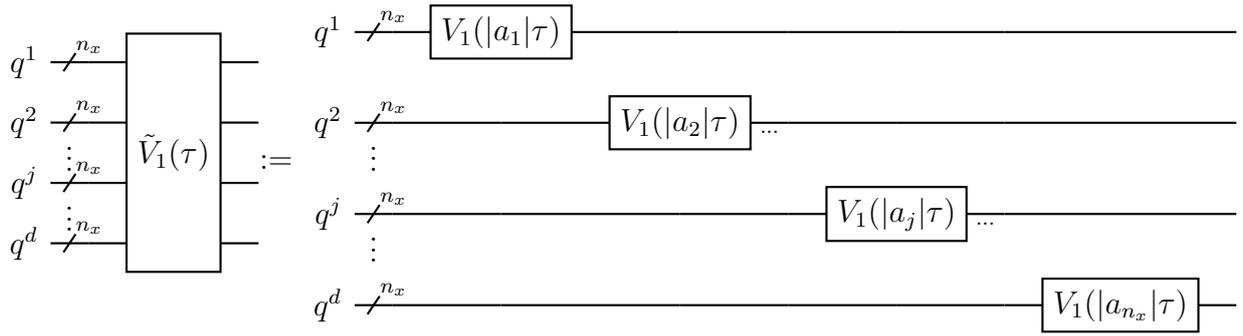

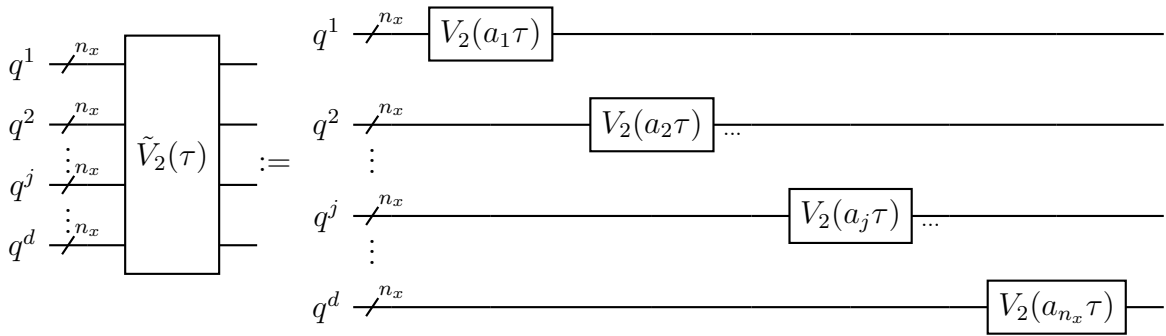
\begin{figure*}[htbp]
    \centering
    \begin{quantikz}[row sep={0.8cm,between origins}]
        \lstick{$q^{1}$} & \qwbundle{n_x} & \gate[4]{\tilde{V}_{2}(\tau)} & \qw \\
        \lstick{$q^{2}$} & \qwbundle{n_x}{\vdots} & & \qw \\
        \lstick{$q^{j}$} & \qwbundle{n_x}{\vdots} & & \qw \\
        \lstick{$q^{d}$} & \qwbundle{n_x} & & \qw \\
    \end{quantikz} 
    :=
    \begin{quantikz}
        \lstick{$q^{1}$} & \qwbundle{n_x} & \gate{V_{2}(a_{1}\tau)} &  \qw & \qw & \qw & \qw & \qw & \qw \\
        \lstick{$q^{2}$} & \qwbundle{n_x}{\vdots}  & \qw & \gate{V_{2}(a_{2}\tau)} & \qw{\ldots} & \qw  & \qw & \qw  & \qw \\
        \lstick{$q^{j}$} & \qwbundle{n_x}{\vdots}  & \qw & \qw & \qw & \gate{V_{2}(a_{j}\tau)} & \qw{\ldots} & \qw & \qw \\
        \lstick{$q^{d}$} & \qwbundle{n_x} & \qw & \qw & \qw & \qw & \qw & \gate{V_{2}(a_{n_x}\tau)} & \qw 
    \end{quantikz}
    \caption{Quantum circuit for $\tilde{V}_2(\tau)$.}
    \label{fig:adv:circuit:tildeV2}
\end{figure*}

\begin{figure*}[htbp]
    \centering
    \begin{quantikz}[row sep={0.8cm,between origins}]
        \lstick{$q$} & \qwbundle{n_x d} & \gate{\tilde{V}_{1}(\tau)} & \gate{\tilde{V}_{1}^{2}(\tau)} & \qw & \ldots & \gate{\tilde{V}_{1}^{2^{n_p-1}}(\tau)} & \gate{\tilde{V}_{1}^{-2^{n_p-1}}(\tau)} & \gate{\tilde{V}_{2}(\tau)} & \qw \\
        \lstick{$k_0$} & \qw & \ctrl{-1} & \qw & \qw & \ldots & \qw & \qw & \qw & \qw \\
        \lstick{$k_1$} & \qw{\vdots} & \qw & \ctrl{-2} & \qw & \ldots & \qw & \qw & \qw & \qw \\
        \lstick{$k_{n_p-1}$} & \qw & \qw & \qw & \qw & \ldots & \ctrl{-3} & \qw & \qw & \qw \\
    \end{quantikz} 
    \caption{Quantum circuit for $V_{\text{adv}}(\tau)$.}
    \label{fig:adv:circuit:V}
\end{figure*}
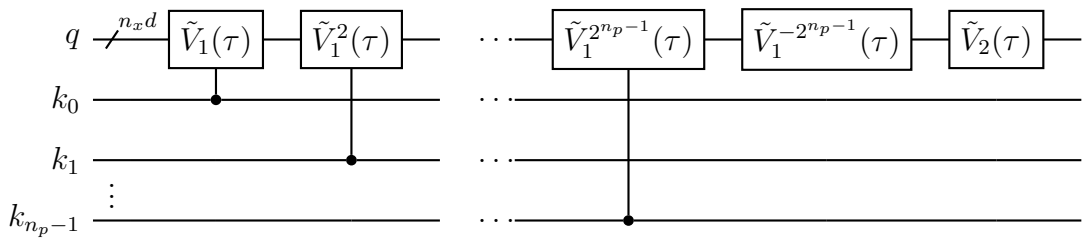

\section{Complexity analysis}
\label{sec:complexity}

In this subsection, we estimate the complexity of the quantum circuits constructed in the previous sections and demonstrate their scalability. 

Typical quantum algorithms for solving partial differential equations often depend on oracles for general matrices, but the cost of preparing these oracles is often not considered, which can make it challenging to assess their feasibility accurately. In our study, we provide complete algorithms for the heat equation and the advection equation (with constant coefficients) that does not require any oracles, making it more practical. The complexity results we present here are based on basic gates, i.e., single-qubit gates and CNOT gates, ensuring clarity and accuracy.

\subsection{The heat equation}
\label{sec:complexity:heat}

\begin{mylemma}\label{lemma:heat:error:tau}
Consider the Schr\"odinger equation $d|\mathbf{u}(t)\rangle /dt = i\mathbf{H}_{\text{heat}}|\mathbf{u}(t)\rangle $, where the Hamiltonian $\mathbf{H}_{\text{heat}}$ is given by Equation \eqref{eqn:heat:Hd}. The time evolution operator $U_{\text{heat}}(\tau)=\exp(i\mathbf{H}_{\text{heat}}\tau)$ with a time increment $\tau$ can be approximated by the unitary $V_{\text{heat}}(\tau)$ in Equation \eqref{eqn:heat:gate:Vtau}, and its explicit circuit implementation is depicted in Figure \ref{fig:heat:circuit:Wj}-\ref{fig:heat:circuit:V}. Furthermore, the approximation error in terms of the operator norm is upper-bounded by
\begin{equation}
    \left\| U_{\text{heat}}(\tau) - V_{\text{heat}}(\tau) \right\| \leq \frac{dN_p \gamma_{0}^2 \tau^2 (n_x-1)}{4},
\end{equation}
and
\begin{equation}
    \left\| U_{\text{heat}}(T) - V_{\text{heat}}^{r}(\tau) \right\| \leq \frac{dN_p \gamma_{0}^2 T^2 (n_x-1)}{4r},
\end{equation}
where $r = T/\tau$, $d$ denotes the spatial dimension, $N_p=2^{n_p}$ and $N_x=2^{n_x}$ represent the number of grid points for the variables $p$ and $x$ (per dimension), respectively. Additionally, $\gamma_{0}$ is defined in Equation \eqref{def:heat:H0gamma}.
\end{mylemma}

\begin{proof}
    According to the theory of the Trotter splitting error with commutator scaling \cite{childs2021theory} and the analysis in \cite{sato2024hamiltonian}, we have
    \begin{equation}\label{eqn:heat:error:U0V0}
    \begin{aligned}
        &\left\| U_0(\tau) - V_0(\tau) \right\| \\
        \leq& \frac{\gamma_{0}^2\tau^2}{2} \sum_{j=1}^{n_x} \sum_{j^\prime=j+1}^{n_x} \left\| \left[(s_{j}^{-}+s_{j}^{+}),  (s_{j^\prime}^{-}+s_{j^\prime}^{+}) \right] \right\| \\
        \leq& \frac{\gamma_{0}^2\tau^2(n_x-1)}{2},
    \end{aligned}
    \end{equation}
    where we have applied the commutator results presented in Appendix \ref{sec:appendix:equations:commutator}, Equation \eqref{eqn:appendix:commutator:s1+}. From Equation \eqref{eqn:appendix:operator:heat}, we have
    \begin{equation}
    \begin{aligned}
        &\left\| U_{\text{heat}}(\tau) - V_{\text{heat}}(\tau) \right\| \\
        \leq& \frac{d N_p}{2} \left\| U_0(\tau) - V_0(\tau) \right\| \\
        \leq& \frac{dN_p \gamma_{0}^2 \tau^2 (n_x-1)}{4}.
    \end{aligned}
    \end{equation}

    \begin{equation}
    \begin{aligned}
        &\left\| U_{\text{heat}}(T) - V_{\text{heat}}^{r}(\tau) \right\| \\
        \leq& r \left\| U_{\text{heat}}(\tau) - V_{\text{heat}}(\tau) \right\| \\
        \leq& \frac{dN_p \gamma_{0}^2 T^2 (n_x-1)}{4r}.
    \end{aligned}
    \end{equation}
\end{proof}

\begin{mylemma}\label{lemma:heat:complexity:tau}
    The approximated time evolution operator $V_{\text{heat}}(\tau)$ in Equation \eqref{eqn:heat:gate:Vtau} can be implemented using $O(dN_p n_x)$ single-qubit gates and at most $\mathcal{Q}_{V_{\text{heat}}} = O(d N_p n_x^2)$ CNOT gates for $n_x \geq 3$.
\end{mylemma}

\begin{proof}
The implementation of $V_{\text{heat}}(\tau)$ involves a maximum of $2^{n_p-1}$ $\tilde{V}_{0}(-\tau)$ gates and $\sum_{m=0}^{n_p-1} 2^{m} = 2^{n_p}-1$ controlled $\tilde{V}_{0}(\tau)$ gates. Each $\tilde{V}_{0}$ gate consists of $d$ $V_{0}$ gates. Furthermore, $V_{0}$ consists of a phase gate $\text{Ph}$ and $W_j, j=1,\dots,n$. $W_j$ can be decomposed into a multi-controlled RZ gate (a controlled RZ gate for $j = 2$), $2$ Hadmard gates and $2(j-1)$ CNOT gates in the case $\lambda=0$. Hence the number of single-qubit gates is $d2^{n_p-1}(2n_x+1) = O(dN_p n_x)$. The total number of CNOT gates in the circuit implementation of $V_{\text{heat}}$ is determined by
\begin{equation}
        \mathcal{Q}_{V_{\text{heat}}} = d2^{n_p-1} \mathcal{Q}_{V_{0}} + d(2^{n_p} - 1) \mathcal{Q}_{c-V_{0}}.
\end{equation}
 According to the decomposition techniques presented in \cite{sato2024hamiltonian,vale2023decomposition}, it is known that a multi-controlled RZ or RX gate with $(j - 1)$ control qubits can be decomposed into single-qubit gates and at most $16j - 40$ CNOT gates. Hence
\begin{equation}
\begin{aligned}
    \mathcal{Q}_{V_{0}} &= 2 + \sum_{j=3}^{n_x} (16j - 40) + \sum_{j=2}^{n_x} 2(j-1) \\
    &= 9n_x^2 - 33n_x + 34.
\end{aligned}
\end{equation}
The controlled $V_{0}$ gate consists of a controlled phase gate and controlled $W_j, j=1,\dots,n$. $c-W_j$ can be decomposed into a multi-controlled RZ gate (a controlled RZ gate for $j = 1$), two controlled H gates, and $2(j-1)$ controlled CNOT gates. Therefore, the total number of gates in the circuit implementation of the controlled $V_{0}$ gate can be upper-bounded by
\begin{equation}
\begin{aligned}
    &\mathcal{Q}_{c-V_{0}} \\
    =& 2 + \sum_{j=3}^{n_x+1} (16j-40) + 2n_x + \sum_{j=2}^{n_x} 2(j-1) * 8 \\
    =& 16n_x^2 - 22n_x + 10.
\end{aligned}
\end{equation}
Then, $V_{\text{heat}}$ can be implemented using single-qubit gates and at most
\begin{equation}
\begin{aligned}
    \mathcal{Q}_{V_{\text{heat}}} &= d2^{n_p-1} \mathcal{Q}_{V_{0}} + d(2^{n_p} - 1) \mathcal{Q}_{c-V_{0}} \\
    &= O(d N_p n_x^2)
\end{aligned}
\end{equation}
CNOT gates.

\end{proof}

\begin{mylemma}\label{lemma:heat:complexity:T}
Let $\mathbf{H}_{\text{heat}}$ be the Hamiltonian defined in Equation \eqref{eqn:heat:Hd}. The time evolution operator $U_{\text{heat}}(T) = \exp(i\mathbf{H}_{\text{heat}}T)$ for a time duration $T$ can be implemented on a $(dn_x + n_p)$-qubit system using quantum circuits with $O(d^2 n_x^2 N_p^2 \gamma_{0}^2 T^2 /\varepsilon)$ single-qubit gates and $O(d^2 n_x^3 N_p^2 \gamma_{0}^2 T^2 /\varepsilon)$ non-local gates, within an additive error of $\varepsilon$. The quantum circuit for $U_{\text{heat}}(T)$ involves repetitive applications of the one-time step unitary $V_{\text{heat}}(\tau)$ depicted in Figure \ref{fig:heat:circuit:V}.
\end{mylemma}

\begin{proof}
To suppress the error in the simulation over the total time $T$ to a small value $\varepsilon$, it is sufficient to divide the total time $T$ into $r$ intervals, where $r = T/\tau$, such that 
\begin{equation}
    \left\| U_{\text{heat}}^r(\tau) - V_{\text{heat}}^r(\tau) \right\| \leq \frac{dN_p \gamma_{0}^2 T^2 (n_x-1)}{4r} \leq \varepsilon,
\end{equation}
which can be rearranged as 
\begin{equation}
        r \geq \frac{dN_p \gamma_{0}^2 T^2 (n_x-1)}{4\varepsilon}.
\end{equation}

According to Lemma \ref{lemma:heat:complexity:tau}, $V_{\text{heat}}(\tau)$ can be implemented using $O(dN_pn_x)$ single-qubit gates and at most $\mathcal{Q}_{V_{\text{heat}}} = O(d N_p n_x^2)$ CNOT gates. Therefore, $V_{\text{heat}}^r(\tau)$ can be implemented using $O(r\mathcal{Q}_{V_{\text{heat}}}) = O(d^2 n_x^3 N_p^2 \gamma_{0}^2 T^2 /\varepsilon)$ non-local gates and $O(d^2 n_x^2 N_p^2 \gamma_{0}^2 T^2 /\varepsilon)$ single-qubit gates.
\end{proof}

\begin{mythm}\label{thm:heat:complexity:solution}
    Given the heat equation  \eqref{mod:heat}, the state $\ket{\mathbf{u}(t)}$, where $\mathbf{u}(t)$ is the solution of Equation \eqref{mod:heat:num} with a mesh size $h$, can be prepared with the precision $\varepsilon$ using the Schr\"odingerization method depicted in Figure \ref{fig:heat:circuit:V} and Figure \ref{fig:circuit:schro}. This preparation can be achieved using at most $\tilde{O}\left( d^2 T^2 \left\| \mathbf{u}(0) \right\|^3 / (\left\| \mathbf{u}(T) \right\|^3 h^4 \varepsilon^3) \right)$ single-qubit gates and CNOT gates.
\end{mythm}

\begin{proof}
As depicted in Figure \ref{fig:circuit:schro}, the quantum circuit involves a quantum Fourier transform, an inverse quantum Fourier transform, a unitary $V_{\text{heat}}^r(\tau)$ operation, and several measurements with $M_{\geq 0}$. The (inverse) quantum Fourier transform can be implemented with $O(n_p^2)$ CR gates, which corresponds to $O(n_p^2)$ CNOT gates \cite{nielsen2010quantum,vale2023decomposition,lin2022lecture}.

The complexity analysis of simulating the unitary $V_{\text{heat}}^r(\tau)$ as concluded in Lemma \ref{lemma:heat:complexity:T} can be applied with $n_x = O(\log(L/h))$ and $\gamma_{0} = a/(h^2R)$. This results in a computational cost of $O\left( d^2 T^2 N_p^2 \log^3(L/h) /(h^4\varepsilon^\prime R^2) \right)$, where $\varepsilon^\prime$ represents the desired precision $\left| U_{\text{heat}}(T) - V_{\text{heat}}^r(\tau) \right| \leq \varepsilon^\prime$.

The number of measurements and the error of the output state can be discussed in the context of the given expressions $\ket{\hat{\mathbf{v}}(T)} = U(T) \ket{\hat{\mathbf{v}}(0)}$ and $\ket{\hat{\mathbf{v}}_{D}(T)} = V_{\text{heat}}^r(\tau) \ket{\hat{\mathbf{v}}(0)}$. To retrieve $\ket{\mathbf{v}(T)}$ and $\ket{\mathbf{v}_{D}(T)}$, an inverse quantum Fourier transform can be performed on these states. After the projection with $M_{\geq 0}$, one can observe that
\begin{equation}
\begin{aligned}
    &\frac{\left\| M_{\geq 0} \mathbf{v}(T) - \mathbf{u}(T) \otimes \mathbf{p}_{\geq 0} \right\|}{\left\| \mathbf{u}(T) \right\| \left\| \mathbf{p}_{\geq 0} \right\| } = O\left( \frac{\pi R}{N_p} + e^{- \pi R} \right)
\end{aligned}
\end{equation}
which represents the (relative) discretization error over the $p$ variable, $\mathbf{p}_{\geq 0} := \sum_{p_k\geq 0} e^{-p_k} \ket{k}$. On the other hand,
\begin{equation}
\begin{aligned}
    &\left\| \ket{M_{\geq 0} \ket{\mathbf{v}_{D}(T)}} -\ket{M_{\geq 0} \ket{\mathbf{v}(T)}} \right\| \\
    = & O\left( \frac{\left\| M_{\geq 0} \left(\ket{\mathbf{v}_{D}(T)} - \ket{\mathbf{v}(T)}\right) \right\|}{\left\| M_{\geq 0} \ket{\mathbf{v}(T)} \right\|} \right)\\
    = & O\left( \frac{\left\| \mathbf{v}(0) \right\|}{\left\| \mathbf{u}(T)\right\| \left\| \mathbf{p}_{\geq 0} \right\|} \left\| U(T) - V_{\text{heat}}^r(\tau) \right\| \right) \\
    =& O\left( \frac{\left\|\mathbf{u}(0)\right\|}{\left\| \mathbf{u}(T) \right\|} \varepsilon^\prime \right),
\end{aligned}
\end{equation}
where we have used the fact that $\left\|\mathbf{p}_{\geq 0}\right\| = O(\left\|\mathbf{p}\right\|)$. To ensure the precision of the output state, we require that $R=O(\log(1/\varepsilon))$, $N_p = O(R/\varepsilon) = \tilde{O}(1/\varepsilon)$ and $\varepsilon^\prime = O\left(\left\| \mathbf{u}(T) \right\| \varepsilon / \left\| \mathbf{u}(0) \right\|\right)$. Note that the probability of getting the desired state is $O(\left\| \mathbf{u}(T) \right\|^2 / \left\| \mathbf{u}(0) \right\|^2)$, hence $O(\left\| \mathbf{u}(0) \right\|^2 / \left\| \mathbf{u}(T) \right\|^2)$ measurements are needed.

Combining all the steps, the Schr\"odingerization method for solving the heat equation utilizes at most $\tilde{O}\left( d^2 T^2 \frac{\left\| \mathbf{u}(0) \right\|^3}{\left\| \mathbf{u}(T) \right\|^3 h^4 \varepsilon^3} \right)$ single-qubit gates and CNOT gates.

\end{proof}

\begin{myremark}\label{remark:heat:thm:compare}
    In the classical implementation, the application of the difference operator requires $O(s2^{n_x d})$ arithmetic operations, where $s$ denotes the sparsity of the difference operator and $s(D_{D}^{\Delta})=3$ in the case of the heat equation. Using the forward Euler scheme, the classical simulation up to time $T$ within the additive error $\varepsilon$ requires $O(T^2/\varepsilon)$ steps, which results in $O(s2^{n_xd}(T^2/\varepsilon))$ arithmetic operations. In addition, the Courant-Friedrichs-Lewy (CFL) condition requires the time step $\tau = O(h^2)$. As a result, classical simulation up to time $T$ within the additive error $\varepsilon$ requires $O(s2^{n_xd}(T^2/\varepsilon + T/h^2)) = O(s(T^2/h^d\varepsilon + T/h^{d+2}))$  arithmetic operations. Therefore, quantum advantage is achievable in higher-dimensional (for example, $d>5$ if $h=O(\varepsilon)$) cases, where the quantum algorithm demands only a logarithmic number of operations compared to classical methods.
\end{myremark}

\begin{myremark}\label{remark:heat:thm:dp-order}
    The discretization over $p$ is essentially a Fourier spectral method applied to the function $e^{-|p|}$. The observed first-order convergence in $p$ arises from the limited regularity of $e^{-|p|}$, which is continuous but not differentiable ($C^1$). To address this, we can consider a more general initial data   $v(0,x,p) = g(p) u(0,x)$, where $g(p)$ is defined on $\mathbb{R}$ and satisfies:
    \begin{equation}
    g(p) = \left\{
    \begin{aligned}
        &h(p), \quad p \in (-\infty, 0], \\
        &e^{-p}, \quad p \in (0, +\infty).
    \end{aligned}  
    \right.
    \end{equation}
    The convergence rate can be enhanced by selecting a smoother $g(p)$, such as higher-order polynomials \cite{jin2024schrodingerisation}. Assuming $g \in C^k(\mathbb{R})$, the discretization error over the $p$ variable improves to $O(\Delta p^{k+1} + e^{-\pi R})$. Consequently, it suffices to choose $N_p = O(R/\varepsilon^{\frac{1}{k+1}}) = \tilde{O}(1/\varepsilon^{\frac{1}{k+1}})$. Overall, the number of single-qubit gates and CNOT gates required for solving the heat equation can be reduced to $\tilde{O}\left( d^2 T^2 \left\| \mathbf{u}(0) \right\|^3 / \left(\left\| \mathbf{u}(T) \right\|^3 h^4 \varepsilon^{1+\frac{2}{k+1}} \right) \right)$.
\end{myremark}

\subsection{The advection equation}
\label{sec:complexity:adv}

In this section, we conduct a complexity analysis of the quantum circuit presented in Section \ref{sec:advection} for the upwind scheme of the advection equation.

\begin{mylemma}\label{lemma:adv:error:tau}
Consider the Schr\"odinger equation $d|\mathbf{u}(t)\rangle /dt = i\mathbf{H}_{\text{adv}}|\mathbf{u}(t)\rangle $, where the Hamiltonian $\mathbf{H}_{\text{adv}}$ is given by Equation \eqref{eqn:adv:Hd}. The time evolution operator $U_{\text{adv}}(\tau)=\exp(i\mathbf{H}_{\text{adv}}\tau)$ with a time increment $\tau$ can be approximated by the unitary $V_{\text{adv}}(\tau)$ in Equation \eqref{eqn:adv:gate:Vtau}, and its explicit circuit implementation is depicted in Figure \ref{fig:adv:circuit:U11}-\ref{fig:adv:circuit:V}. Furthermore, the approximation error in terms of the operator norm is upper-bounded by
\begin{equation}
\begin{aligned}
    &\left\| U_{\text{adv}}(\tau) - V_{\text{adv}}(\tau) \right\| \\
    \leq & \frac{\tau^2 n_x (N_p\gamma_2^2 + 2N_p\gamma_1\gamma_2 + 2\gamma_1^2)}{4} \sum_{\alpha=1}^{d} a_{\alpha}^2,
\end{aligned}
\end{equation}
and
\begin{equation}
\begin{aligned}
    &\left\| U_{\text{adv}}(T) - V_{\text{adv}}^r(\tau) \right\| \\
    \leq & \frac{T^2 n_x (N_p\gamma_2^2 + 2N_p\gamma_1\gamma_2 + 2\gamma_1^2)}{4r} \sum_{\alpha=1}^{d} a_{\alpha}^2,
\end{aligned}
\end{equation}
where $r = T/\tau$, $d$ denotes the spatial dimension, $N_p=2^{n_p}$ and $N_x=2^{n_x}$ represent the number of grid points for the variables $p$ and $x$, respectively. Additionally, $\gamma_{1}$, $\gamma_{2}$ are defined in Equation \eqref{def:adv:H1gamma}, Equation \eqref{def:adv:H2gamma}.
\end{mylemma}

\begin{proof}    
    Similarly to the proof of Lemma \ref{lemma:heat:error:tau}, the Trotter splitting error is upper-bounded by the norm of commutators. We utilize several commutator results here. For detailed calculations, please refer to Appendix \ref{sec:appendix:equations}. Recalling the definition of $U_{1}(\tau)$, $U_{2}(\tau)$, $V_{1}(\tau)$ and $V_{2}(\tau)$ in Equation \eqref{eqn:adv:unitary:U1U2} and Equation \eqref{eqn:adv:gate:V1V2}, we have 
    \begin{equation}
    \begin{aligned}
        \left\| U_{1}(\tau) - V_{1}(\tau) \right\| &\leq \frac{\gamma_{1}^2 \tau^2 n_x}{2}, \\
        \left\| U_{2}(\tau) - V_{2}(\tau) \right\| &\leq \frac{\gamma_{2}^2 \tau^2 n_x}{2},
    \end{aligned}
    \end{equation}
    where the details are presented in Equation \eqref{eqn:appendix:operator:U1V1} and Equation \eqref{eqn:appendix:operator:U2V2}. This implies that
    \begin{equation}
    \begin{aligned}
        & \left\| \tilde{U}_{1}(\tau) - \tilde{V}_{1}(\tau) \right\| \leq \frac{\gamma_{1}^2 \tau^2 n_x }{2} \sum_{\alpha=1}^{d} a_{\alpha}^2, \\
        & \left\| \tilde{U}_{2}(\tau) - \tilde{V}_{2}(\tau) \right\| \leq \frac{\gamma_{2}^2 \tau^2 n_x }{2} \sum_{\alpha=1}^{d} a_{\alpha}^2.
    \end{aligned}
    \end{equation}

    In all, we obtain
    \begin{equation}
    \begin{aligned}
        &\left\| U_{\text{adv}}(\tau) - V_{\text{adv}}(\tau) \right\| \\
        \leq & \left\| U_{\text{adv}}(\tau) - U_{*}(\tau) \right\| + \left\| U_{*}(\tau) - V_{\text{adv}}(\tau) \right\| \\
        \leq & \frac{\tau^2}{2} \left\| \left[ \mathbf{A}_{1}\otimes D_{\eta}, \mathbf{A}_{2}\otimes I^{\otimes n_p} \right] \right\| \\
        &+  \frac{N_{p}}{2} \left\| \tilde{U}_{1}(\tau) - \tilde{V}_{1}(\tau) \right\| +  \left\| \tilde{U}_{2}(\tau) - \tilde{V}_{2}(\tau) \right\| \\
        \leq & \frac{\tau^2 N_p \gamma_{1}\gamma_{2} n_x}{2} \sum_{\alpha=1}^{d} a_{\alpha}^2 \\
        &+ \frac{N_p \gamma_{1}^2 \tau^2 n_x }{4} \sum_{\alpha=1}^{d} a_{\alpha}^2 + \frac{\gamma_{1}^2 \tau^2 n_x }{2} \sum_{\alpha=1}^{d} a_{\alpha}^2 \\
        \leq & \frac{\tau^2 n_x (N_p\gamma_1^2 + 2N_p\gamma_1\gamma_2 + 2\gamma_2^2)}{4} \sum_{\alpha=1}^{d} a_{\alpha}^2,
    \end{aligned}
    \end{equation}
    where we have put Equation \eqref{eqn:appendix:commutator:H1H2} in. Therefore
    \begin{equation}
    \begin{aligned}
        &\left\| U_{\text{adv}}(T) - V_{\text{adv}}^r(\tau) \right\| \\
        \leq & r \left\| U_{\text{adv}}(\tau) - V_{\text{adv}}(\tau) \right\| \\
        \leq & \frac{T^2 n_x (N_p\gamma_1^2 + 2N_p\gamma_1\gamma_2 + 2\gamma_2^2)}{4r} \sum_{\alpha=1}^{d} a_{\alpha}^2.
    \end{aligned}
    \end{equation}
\end{proof}

\begin{mylemma}\label{lemma:adv:complexity:tau}
    The approximated time evolution operator $V_{\text{adv}}(\tau)$ in Equation \eqref{eqn:adv:gate:Vtau} can be implemented using $O(dN_pn_x)$ single-qubit gates and at most $\mathcal{Q}_{V_{\text{adv}}} = O(d N_p n_x^2)$ CNOT gates for $n_x \geq 3$.
\end{mylemma}

\begin{proof}
The proof is similar to the proof of Lemma \ref{lemma:heat:complexity:tau}. The implementation of $V_{\text{adv}}(\tau)$ involves a maximum of a $\tilde{V}_{2}(\tau)$ gate, $2^{n_p-1}$ $\tilde{V}_{1}(-\tau)$ gates and $\sum_{m=0}^{n_p-1} 2^{m} = 2^{n_p}-1$ controlled $\tilde{V}_{1}(\tau)$ gates. Each $\tilde{V}_{1}$ ($\tilde{V}_{2}$) gate consists of $d$ $V_{1}$ ($V_{2}$) gates. Furthermore, $V_{1}$ and $V_{2}$ are similar to $V_0$ but incorporate additional gates $U_{1}^{(1)}$ and $U_{2}^{(1)}$, each of which can be decomposed into a multi-controlled RZ gate, $2$ Hadamard gates, $2$ phase gates and $2(n_x-1)$ $X$ gates and $2(n_x-1)$ CNOT gates. The controlled $U_{1}^{(1)}$ requires $2(n_x-1)$ extra controlled $X$ gates. Hence the number of single-qubit gates is $d(4n_x + 2 + 2 + 2(n_x - 1)) + d 2^{n_p-1} (2n_x + 2 + 2(n_x-1) + 1) = O(dN_pn_x)$.

The total number of CNOT gates in the circuit implementation of $V_{\text{adv}}$ is determined by
\begin{equation}
    \mathcal{Q}_{V_{\text{adv}}} = d \mathcal{Q}_{V_{2}} + d2^{n_p-1} \mathcal{Q}_{V_{1}} + d(2^{n_p} - 1) \mathcal{Q}_{c-V_{1}},
\end{equation}
with
\begin{equation}
\begin{aligned}
    \mathcal{Q}_{V_{2}} = \mathcal{Q}_{V_{1}} &= \mathcal{Q}_{V_{0}} + (16 n_x - 40) + 2(n_x-1) \\
    &= 9n_x^2 - 15n_x -8.
\end{aligned}
\end{equation}
\begin{equation}
\begin{aligned}
    &\mathcal{Q}_{c-V_{1}} \\
    =& \mathcal{Q}_{c-V_{0}} + 16 (n_x+1) - 40 + 2 n_x + 16(n_x-1) \\
    =& 16n_x^2 - 2n_x - 30.
\end{aligned}
\end{equation}
Then, $V_{\text{adv}}$ can be implemented using $O(dN_pn_x)$ single-qubit gates and at most
\begin{equation}
\begin{aligned}
    \mathcal{Q}_{V_{\text{adv}}} &= d \mathcal{Q}_{V_{2}} + d2^{n_p-1} \mathcal{Q}_{V_{1}} + d(2^{n_p} - 1) \mathcal{Q}_{c-V_{1}} \\
    &= O(d N_p n_x^2),
\end{aligned}
\end{equation}
CNOT gates.

\end{proof}

\begin{mylemma}\label{lemma:adv:complexity:T}
Let $\mathbf{H}_{\text{adv}}$ be the Hamiltonian defined in Equation \eqref{eqn:adv:Hd}. The time evolution operator $U_{\text{adv}}(T) = \exp(i\mathbf{H}_{\text{adv}}T)$ for a time duration $T$ can be implemented on a $(dn_x + n_p)$-qubit system using quantum circuits with $O(d n_x^2 T^2 N_p (N_p\gamma_1^2 + 2N_p\gamma_1\gamma_2 + 2\gamma_2^2)\sum_{\alpha=1}^{d} a_{\alpha}^2 /\varepsilon)$ single-qubit gates and $O(d n_x^3 T^2 N_p (N_p\gamma_1^2 + 2N_p\gamma_1\gamma_2 + 2\gamma_2^2)\sum_{\alpha=1}^{d} a_{\alpha}^2 /\varepsilon)$ non-local gates, within an additive error of $\varepsilon$. The quantum circuit for $U_{\text{adv}}(T)$ involves repetitive applications of the one-time step unitary $V_{\text{adv}}(\tau)$ depicted in Figure \ref{fig:adv:circuit:V}.
\end{mylemma}

\begin{proof}
To suppress the error in the simulation over the total time $T$ to a small value $\varepsilon$, it is sufficient to divide the total time $T$ into $r$ intervals, where $r = T/\tau$, such that
\begin{equation}
\begin{aligned}
    & \left\| U_{\text{adv}}^r(\tau) - V_{\text{adv}}^r(\tau) \right\| \\
    \leq & \frac{T^2 n_x (N_p\gamma_1^2 + 2N_p\gamma_1\gamma_2 + 2\gamma_2^2)}{4 r} \sum_{\alpha=1}^{d} a_{\alpha}^2 \\
    \leq & \varepsilon,
\end{aligned}
\end{equation}
which can be rearranged as 
\begin{equation}
        r \geq \frac{T^2 n_x (N_p\gamma_1^2 + 2N_p\gamma_1\gamma_2 + 2\gamma_2^2)}{4\varepsilon} \sum_{\alpha=1}^{d} a_{\alpha}^2 .
\end{equation}

According to Lemma \ref{lemma:adv:complexity:tau}, $V_{\text{adv}}(\tau)$ can be implemented using $O(d N_p n_x)$ single-qubit gates and at most $\mathcal{Q}_{V_{\text{adv}}} = O(d N_p n_x^2)$ CNOT gates. Therefore, $V_{\text{adv}}^r(\tau)$ can be implemented using $O(r\mathcal{Q}_{V_{\text{adv}}}) = O(d n_x^3 T^2 N_p (N_p\gamma_1^2 + 2N_p\gamma_1\gamma_2 + 2\gamma_2^2) \sum_{\alpha=1}^{d} a_{\alpha}^2 /\varepsilon)$ non-local gates and $O(d n_x^2 T^2 N_p (N_p\gamma_1^2 + 2N_p\gamma_1\gamma_2 + 2\gamma_2^2)\sum_{\alpha=1}^{d} a_{\alpha}^2 /\varepsilon)$ single-qubit gates.
\end{proof}

\begin{mythm}\label{thm:adv:complexity:solution}
    Given the advection equation \eqref{mod:adv}, the state $\ket{\mathbf{u}(t)}$, where $\mathbf{u}(t)$ is the solution of Equation \eqref{mod:adv:num} with a mesh size $h$, can be prepared up to precision $\varepsilon$ using the Schr\"odingerization method depicted in Figure \ref{fig:adv:circuit:V} and Figure \ref{fig:circuit:schro}. This preparation can be achieved using at most $\tilde{O}\left( d T^2 \sum_{\alpha=1}^{d} a_{\alpha}^2 \left\| \mathbf{u}(0) \right\|^3 / (\left\| \mathbf{u}(T) \right\|^3 h^2 \varepsilon^3) \right)$ single-qubit gates and CNOT gates.
\end{mythm}

\begin{proof}
The complexity analysis of simulating the unitary $V_{\text{adv}}^r(\tau)$ as concluded in Lemma \ref{lemma:adv:complexity:T} can be applied with $n_x = O(\log(L/h))$ and $\gamma_{1}=1/(2hR)$, $\gamma_{2}=1/(2h)$. This results in a computational cost of $O\left( d^2 T^2 N_p^2 \log^3(L/h) \sum_{\alpha=1}^{d} a_{\alpha}^2 /(h^2\varepsilon^\prime R^2) \right)$, with $\varepsilon^\prime = \left\| \mathbf{u}(T) \right\| \varepsilon / \left\| \mathbf{u}(0) \right\|$.

The number of measurements and the error analysis of the output state are exactly the same as in the proof of Theorem \ref{thm:heat:complexity:solution}, leading to $R=O(\log(1/\varepsilon))$, $N_p = O(R/\varepsilon) = \tilde{O}(1/\varepsilon)$. Combining all the steps, the Schr\"odingerization method for the upwind scheme of the advection equation utilizes at most $\tilde{O}\left( d T^2 \sum_{\alpha=1}^{d} a_{\alpha}^2 \frac{\left\| \mathbf{u}(0) \right\|^3}{\left\| \mathbf{u}(T) \right\|^3 h^2 \varepsilon^3} \right)$ single-qubit gates and CNOT gates.

\end{proof}

\begin{myremark}\label{remark:adv:thm:compare}
    As discussed in Remark \ref{remark:heat:thm:compare}, classical simulation up to time $T$ within the additive error $\varepsilon$ requires $O(s2^{n_xd}(T^2/\varepsilon + T/h)) = O(s(T^2/h^d\varepsilon + T/h^{d+1}))$ arithmetic operations since $\tau=O(h)$ for the advection equation. Therefore, quantum advantage is achievable in higher-dimensional (for example, $d>4$ if $h=O(\varepsilon)$) cases. Compared to the quantum circuit constructed in \cite{sato2024hamiltonian}, where the central difference operator is used to obtain a Hermitian matrix, the paramether $N_p^2$ in the gate complexity comes from the largest value of the Fourier variable $\eta$. However, our method is applicable to general non-Hermitian operators.
\end{myremark}

\begin{myremark}\label{remark:adv:thm:dp-order}
    The improvements in the discretization error over the $p$ variable can be achieved as discussed in Remark \ref{remark:heat:thm:dp-order}. By selecting $ g \in C^k(\mathbb{R}) $, the number of single-qubit gates and CNOT gates required for solving the advection equation can be reduced to $ \tilde{O}\left( d T^2 \sum_{\alpha=1}^{d} a_{\alpha}^2 \left\| \mathbf{u}(0) \right\|^3 / \left(\left\| \mathbf{u}(T) \right\|^3 h^2 \varepsilon^{1+\frac{2}{k+1}} \right) \right)$.
\end{myremark}

\section{Numerical experiments}
\label{sec:numerical}

This section presents the numerical experiments conducted to validate the Schr\"odingerisation circuits introduced in Sections \ref{sec:heat} and \ref{sec:advection}. The experiments are performed using the {\it Qiskit} package \cite{Qiskit}, which facilitates the construction and execution of the quantum circuits on a local simulator. Throughout the implementation phase, the \emph{Statevector} simulator is utilized to retrieve the complete solution, while the \emph{Estimator} function is employed for the measurement of observables. To contextualize the performance of the quantum simulations, comparisons are made against two reference points: a classical implementation of the Schr\"odingerisation method executed using the {\it Scipy} package \cite{2020SciPy-NMeth}, and the direct application of the matrix exponential operator $e^{\mathbf{A} t}$ to the initial condition vector $\mathbf{u}(0)$, achieved through conventional matrix-vector multiplication.

\subsection{The heat equation}
We first conduct a Hamiltonian simulation to solve the one-dimensional heat equation subject to Dirichlet boundary conditions, as explicated in Section \ref{sec:heat},
\begin{equation}
    \left\{ \begin{aligned}
        \partial_t u(t,x) &= a \partial_{xx} u(t,x), \\
        u(0,x) &= f(x),
    \end{aligned}\right. \quad x\in \Omega = [0,L].
\end{equation}
Given the coefficient $a$ and the initial condition $f(x)$ as 
\begin{equation}
a = L/\pi^2, \quad f(x) = \sin(\pi x / L),
\end{equation}
the exact solution is
\begin{equation}
    u(t,x) = e^{-t/L} \sin(\pi x/L).
\end{equation}
Employing the central difference approximation  in conjunction with the Schr\"odingerisation procedure in a sequential manner, the discretized heat equation is governed by the following ODEs
\begin{equation}
    \frac{d \mathbf{u}(t)}{dt} = \mathbf{A} \mathbf{u}(t) := a D_{D}^{\Delta} \mathbf{u}(t),
\end{equation}
\begin{equation}
    \frac{d \hat{\mathbf{v}}(t)}{dt} = i \mathbf{H}_{\text{heat}} \hat{\mathbf{v}}(t) := i \left( \mathbf{A}\otimes D_{\eta} \right) \hat{\mathbf{v}}(t),
\end{equation}
where $\mathbf{u}(t) \in \mathbb{R}^{2^{n_x}}$, $\hat{\mathbf{v}}(t) \in \mathbb{R}^{2^{n_x + n_p}}$. To formulate the problem, we set $L = 17$, $n_x=4$ for the central difference approximation, and set $n_p=3,5,7$, $R=4$ for the Schr\"odingerisation method. As shown in Figure \ref{fig:circuit:schro}, $V(\tau)$ is implemented $r=T/\tau$ times with the time increment parameter  $\tau = 0.005$.

The numerical results at time $T=5$ are illustrated in Figure \ref{fig:heat:solution} and Figure \ref{fig:heat:norm}. As shown in Figure \ref{fig:heat:solution}, there is a clear concordance between the numerical solutions derived from the quantum circuits and those obtained via the classical implementation of the Schr\"odingerisation method, regardless of the value of $n_p$. Furthermore, as the number of ancilla qubits $n_p$ increases, the solutions produced by the Schr\"odingerisation method progressively converge towards the solution computed directly using the matrix exponential operator $e^{\mathbf{A} t}$. This convergence underscores the Schr\"odingerisation method's capability to effectively approximate the solution. 
The discrepancy with the exact solution is due to the too small number of mesh points used (in both $x$ and $p$) which can be reduced if more grid points are used, which is not done here due to our limited quantum computing resource.

Figure \ref{fig:heat:norm} illustrates the energy $\left\| \mathbf{u} \right\|^2$ obtained from the circuits using $10000$ shots. This result is achieved using two approaches:
\begin{itemize}
    \item Measure the last qubit in register $p$ to obtain the probability $P(1)$ of it being in state $\ket{1}$. This corresponds to the projection of the state $\ket{\mathbf{v}(T)}$ onto $\ket{\mathbf{u}(T)}\ket{k}$ for $ k \geq \frac{N_p}{2}$, i.e., $p_k \geq 0$. Thus, $P(1)$ provides an estimate of the magnitude of $M_{\geq 0} \ket{\mathbf{v}(T)}$, and $\left\| \mathbf{u}(T) \right\|$ is obtained by:
    \begin{equation}
        \left\| \mathbf{u}(T) \right\|^2 \approx \frac{P(1) \left\| \mathbf{u}(0)\right\|^2 \left\| \mathbf{p} \right\|^2}{\left\| \mathbf{p}_{\geq 0} \right\|^2}.
    \end{equation}

    \item Measure the quantum state $|\mathbf{v}\rangle$ with the observable
    \begin{equation}\label{eqn:norm:observable}
        \text{obs} := I^{\otimes n_x} \otimes \ket{1}\bra{1} \otimes \ket{0}\bra{0}^{\otimes (n_p-1)},
    \end{equation}
    which projects the state $|\mathbf{v}(T)\rangle$ to $\ket{\mathbf{u}(T)} \ket{\frac{N_p}{2}}, p_{N_p/2}=0$. Note that in this case $\ket{\mathbf{u}(T)}$ is recovered with $p=0$ since $\mathbf{A}$ is negative definite. For general ODE systems, Jin, Liu and Ma systematically studied the important issue of recovering the original variables in \cite{jin2024schr}.
\end{itemize}
The numerical results validate that both estimates reliably capture the energy measurements. The code is available at  \href{https://github.com/hjp3268/Schrodingerisation-Circuits}{GitHub} \cite{circuitcode}.

\begin{figure*}[htbp]
    \centering
    \includegraphics[width=\textwidth]{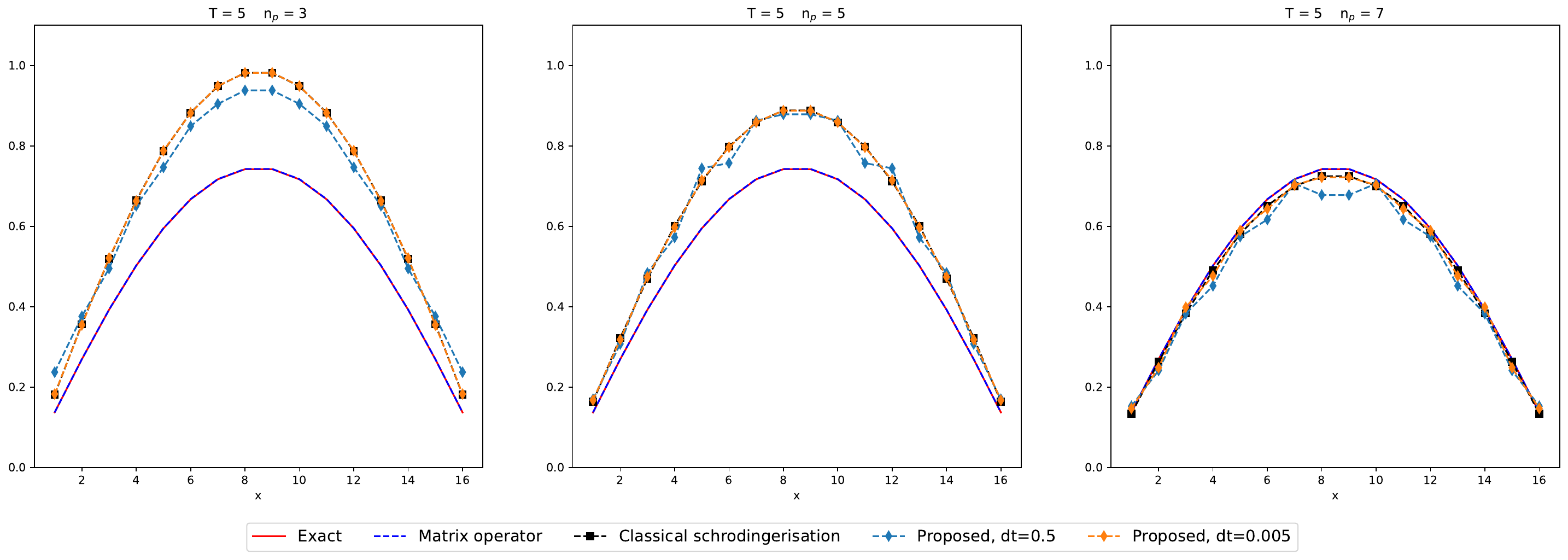}
    \caption{Numerical solutions $u$ of the heat equation using the matrix exponential operator $e^{\mathbf{A}T}$, the classical implementation and the proposed quantum circuits ($n_p=3,5,7$, respectively) of the Schr\"odingerisation method.}
    \label{fig:heat:solution}
\end{figure*}

\begin{figure*}[htbp]
    \centering
    \includegraphics[width=0.6\textwidth]{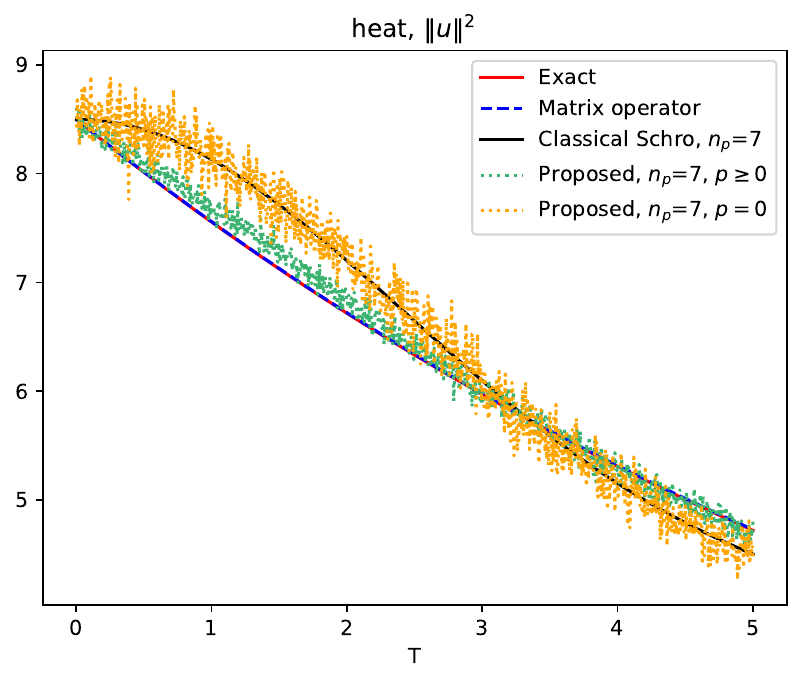}
    \caption{The energy $\left\| u \right\|^2$ of the heat equation using the matrix exponential operator $e^{\mathbf{A}T}$, the classical implementation of the Schr\"odingerisation method and observations of the proposed quantum circuits ($n_p=7$) using $10000$ shots.}
    \label{fig:heat:norm}
\end{figure*}

\subsection{The advection equation}

Next, we consider the one-dimensional advection equation subject to periodic boundary conditions, as explicated in Section \ref{sec:advection},
\begin{equation}
    \left\{ \begin{aligned}
        \partial_t u(t,x) &= a \partial_{x} u(t,x), \\
        u(0,x) &= f(x),
    \end{aligned}\right. \quad x\in \Omega = [0,L].
\end{equation}
Given the coefficient $a$ and the initial condition $f(x)$ as 
\begin{equation}
a = 1, \  f(x) = \left\{ 
\begin{aligned}
    0, \quad x \in \left[kL, kL+L/2\right), \\
    1, \quad x \in \left[kL+L/2, kL\right),
\end{aligned} \right.
\end{equation}
the exact solution is
\begin{equation}
    u(t,x) = f(x+at).
\end{equation}
As discussed above, we use the upwind scheme to avoid numerical oscillations across discontinuities, and then the Schr\"odingerisation procedure is applied:
\begin{equation}
    \frac{d \mathbf{u}(t)}{dt} = \mathbf{A} \mathbf{u}(t) := a D_{P}^{+} \mathbf{u}(t),
\end{equation}
\begin{equation}
\begin{aligned}
    \frac{d \hat{\mathbf{v}}(t)}{dt} &= i \left( \mathbf{A}_{1}\otimes D_{\eta} + \mathbf{A}_{2}\otimes I^{\otimes n_p} \right) \hat{\mathbf{v}}(t) \\
    &\triangleq i \mathbf{H}_{\text{adv}} \hat{\mathbf{v}}(t),
\end{aligned}
\end{equation}
where $\mathbf{u}(t) \in \mathbb{R}^{2^{n_x}}$, $\hat{\mathbf{v}}(t) \in \mathbb{R}^{2^{n_x + n_p}}$. To formulate the problem, we set $L = 16$, $n_x=4$ for the upwind scheme, $n_p=3,5,7$, $R=4$ for the Schr\"odingerisation method and the time increment parameter $\tau = 0.005$ for the quantum circuits.

The numerical results at time $T=3$ are illustrated in Figure \ref{fig:adv:solution} and Figure \ref{fig:adv:norm}. As observed previously, the numerical solutions generated from the quantum circuits agree well with those acquired through the classical Schr\"odingerisation implementation, and demonstrate convergence towards the solution computed directly using $e^{\mathbf{A} t}$ as $n_p$ increases. In contrast to the numerical results obtained via the central difference method described in \cite{sato2024hamiltonian}, the solutions derived through this approach are non-oscillatory due to the dissipative nature of the upwind scheme, rather than oscillatory for the central scheme in \cite{sato2024hamiltonian}, across discontinuities.

Figure \ref{fig:adv:norm} illustrates the energy $\left\| \mathbf{u} \right\|^2$ observed from the circuits using $10000$ shots. This is achieved by measuring the last qubit in register $p$ and measuring the quantum state $|\mathbf{v}\rangle$ with the observable \eqref{eqn:norm:observable}. The numerical results closely approximate those derived from the matrix exponential operator and exhibit dissipative characteristics attributable to the properties of the upwind scheme.

\begin{figure*}[htbp]
    \centering
    \includegraphics[width=\textwidth]{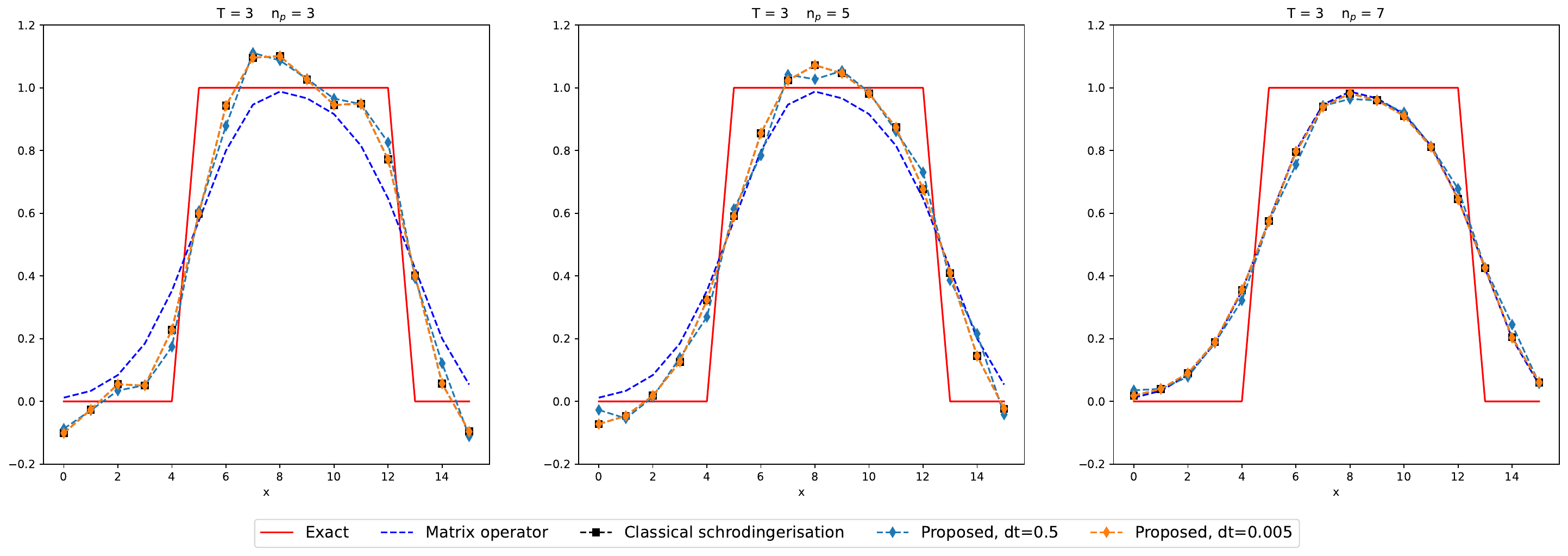}
    \caption{Numerical solutions $u$ of the advection equation using the matrix exponential operator $e^{\mathbf{A}T}$, the classical implementation and the proposed quantum circuits ($n_p=3,5,7$, respectively) of the Schr\"odingerisation method.}
    \label{fig:adv:solution}
\end{figure*}

\begin{figure*}[htbp]
    \centering
    \includegraphics[width=0.6\textwidth]{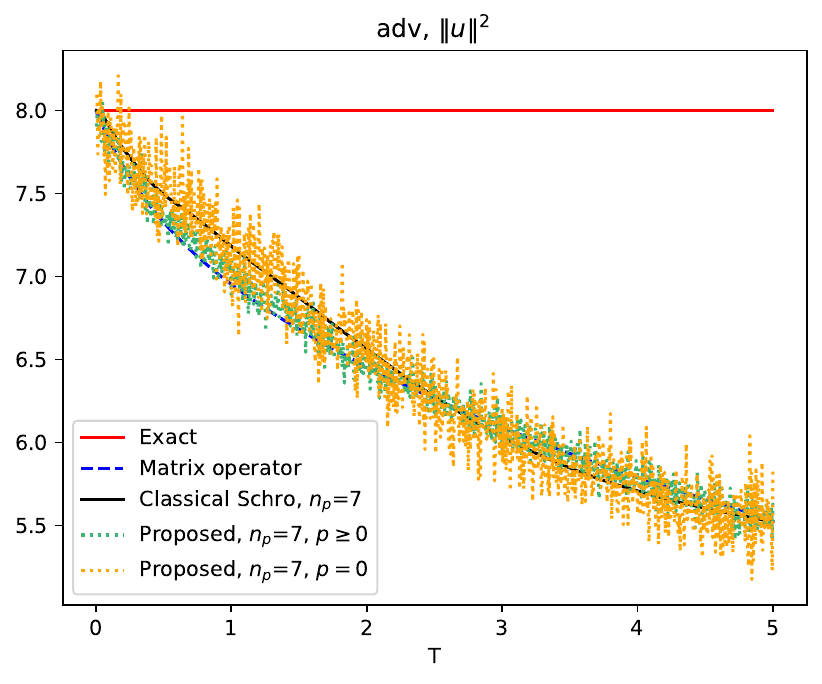}
    \caption{The energy $\left\| u \right\|^2$ of the advection equation using the matrix exponential operator $e^{\mathbf{A}T}$, the classical implementation of the Schr\"odingerisation method and observations of the proposed quantum circuits ($n_p=7$) using $10000$ shots.}
    \label{fig:adv:norm}
\end{figure*}

\section{Conclusion}
\label{sec:conclusion}

We present a practical implementation of quantum circuits for the Schr\"odingerization algorithm, which enables the solution of general partial differential equations (PDEs). To illustrate the approach, we demonstrate its application to two simple examples: the heat equation and the advection equation with the upwind scheme. Detailed quantum circuits are given and their complexities are analyzed, showing quantum advantages in high dimensions over their classical counterparts. The numerical experiments illustrate that the results obtained from these quantum circuits align closely with those achieved through the classical implementation of the Schr\"odingerisation method, as well as with the exact solution of the ordinary differential equation (ODE) discretized from the original partial differential equation (PDE).

In our future research, we aim to expand upon this work by developing scalable quantum circuits to accommodate higher order schemes, diverse boundary and interface conditions, and more challenging PDEs such as Maxwell and radiative transfer equations. Furthermore, we plan to implement these circuits on real quantum computers, thus bridging the gap between theory and practical quantum computing applications.

\section*{Acknowledgement}

SJ, NL, and LZ were supported by NSFC grant No. 12341104, the Shanghai Jiao Tong University 2030 Initiative, and the Fundamental Research Funds for the Central Universities. SJ was also supported by the Shanghai Municipal Science and Technology Major Project (2021SHZDZX0102),  the Innovation Program of Shanghai Municipal Education Commission (No. 2021-01-07-00-02-E00087). NL was also supported by NSFC grant No. 12471411 and the Science and Technology Program of Shanghai, China (21JC1402900). LZ was also supported by NSFC grant No. 12271360, Shanghai Municipal Science and Technology Project (22JC1401600). 

\bibliographystyle{quantum}
\bibliography{ref}

\onecolumn\newpage
\appendix

\section{Several proofs related to the derivations of the circuits.}
\label{sec:appendix:circuit}

\subsection{Proof of Lemma \ref{lemma:evolution:bell}}
\label{sec:appendix:circuit:lemma:bell}
Direct calculation of $e^{i\lambda} s_{j}^{-} + e^{-i\lambda} s_{j}^{+}$ yields:
\begin{equation}\label{eqn:origin:sj}
    e^{i\lambda} s_{j}^{-} + e^{-i\lambda} s_{j}^{+} = I^{\otimes (n_x-j)} \otimes \left( e^{i\lambda} \ket{a_j}\bra{b_j} + e^{-i\lambda} \ket{b_j}\bra{a_j} \right),
\end{equation}
with
\begin{equation}
    \ket{a_j} := |0\rangle |1\rangle^{\otimes(j-1)}, \quad \ket{b_j} := |1\rangle |0\rangle^{\otimes(j-1)}.
\end{equation}
Recalling Lemma \ref{lemma:changebasis}, one can define the unitary matrix $B_j(\lambda)$ such that
\begin{equation}
    B_{j}(\lambda) \ket{0}\ket{1}^{\otimes (j-1)} = \frac{\ket{a} + e^{-i\lambda}\ket{b}}{\sqrt{2}}, \quad 
    B_{j}(\lambda) \ket{1}\ket{1}^{\otimes (j-1)} = \frac{\ket{a} - e^{-i\lambda}\ket{b}}{\sqrt{2}},
\end{equation}
which are called the Bell basis in \cite{sato2024hamiltonian}. And $B_{j}(\lambda)$ is constructed by
\begin{equation}\label{eqn:changebasis:Blambda}
    B_j(\lambda):=\left( \prod_{m=1}^{j-1} \text{CNOT}_{m}^{j} \right) P_{j}(-\lambda) H_{j},
\end{equation}
where $H_{j}$ is the Hadamard gate acting on the $j$-th qubit, $P_{j}(\lambda)$ is the Phase gate acting on the $j$-th qubit as
\begin{equation}
    P_{j}(\lambda) := \begin{bmatrix}
        1 & 0 \\ 0 & e^{i\lambda}
    \end{bmatrix},
\end{equation}
and $\text{CNOT}_{m}^{j}$ is the CNOT gate acting on the $m$-th qubit controlled by the $j$-th qubit. Then Equation \eqref{eqn:origin:sj} can be simplified as
\begin{equation}
    e^{i\lambda} s_{j}^{-} + e^{-i\lambda} s_{j}^{+} = I^{\otimes(n_x-j)} \otimes B_{j}(\lambda) \left( Z \otimes |1\rangle\langle 1|^{\otimes(j-1)} \right) B_{j}(\lambda)^{\dagger}.
\end{equation}
The time evolution operator is formulated as
\begin{equation}\label{eqn:unitary:sj}
\begin{aligned}
    \exp \left( i\gamma\tau (e^{i\lambda} s_{j}^{-} + e^{-i\lambda} s_{j}^{+}) \right) 
    &= I^{\otimes(n_x-j)} \otimes B_{j}(\lambda) \text{CRZ}_{j}^{1,\dots,j-1}(-2\gamma\tau) B_{j}(\lambda)^{\dagger} \\
    &\triangleq I^{\otimes(n_x-j)} \otimes W_{j}(\gamma\tau, \lambda),
\end{aligned}
\end{equation}
where $\text{CRZ}_{j}^{1,\dots,j-1}(\theta)$ is the multi-controled RZ gate as defined in Lemma \ref{lemma:changebasis}.

\subsection{Proof of Lemma \ref{lemma:H1:select}}
\label{sec:appendix:circuit:lemma:H1:select}
We utilize the binary representation of integers $k = (k_{n_p-1}\cdots k_0.) = \sum_{m=0}^{n_p-1} k_{m} 2^m$. This allows us to express $Q_0^{k}(\tau)$ as $\prod_{m=0}^{n_p-1} Q_{0}^{k_{m} 2^m}(\tau)$, resulting in a reduced number of gates:
\begin{equation}\label{eqn:min:implement}
\begin{aligned}
    Q=&\sum_{k=0}^{N_p-1} Q_0^{k}(\tau) \otimes |k\rangle \langle k| \\
    =& \sum_{k_{n_p-1},\dots,k_{0}} \prod_{m=0}^{n_p-1} Q_{0}^{k_{m} 2^m}(\tau) \otimes \left(\ket{k_{n_p-1}}\bra{k_{n_p-1}}\right) \otimes \cdots \otimes (\ket{k_0}\bra{k_{0}}) \\
    =& \sideset{}{'}\prod_{m=0}^{n_p-1} \left( Q_{0}^{k_{m} 2^m}(\tau) \otimes \sum_{k_{m}} |k_m\rangle \langle k_m| \right) \\
    =& \sideset{}{'}\prod_{m=0}^{n_p-1} \left( Q_0^{2^m}(\tau) \otimes |1\rangle \langle 1| + I^{\otimes n_x} \otimes |0\rangle \langle 0| \right),
\end{aligned}
\end{equation}
where the primed product $\sideset{}{'}\prod$ denotes the regular matrices product for the first register (consisting of $n_x$ qubits) and the tensor product for the second register (consisting of $n_p$ qubits).

\subsection{Proof of Equation \eqref{eqn:heat:notation}}
\label{sec:appendix:circuit:heat:notation}
By definition, 
\begin{equation}
    U_{\text{heat}}(\tau) := \exp(i\mathbf{H}_{\text{heat}}\tau) = \exp \left( i \tau \sum_{k=0}^{N_p-1} \left(k-\frac{N_p}{2}\right) \sum_{\alpha=1}^{d} \left(\mathbf{H}_{0}\right)_{\alpha} \otimes |k\rangle \langle k| \right).
\end{equation}
Combining with the following properties,
\begin{equation}
\begin{aligned}
    &\exp\left( \sum_{k} A_k \otimes |k\rangle \langle k| \right) = \sum_{k} \exp(A_k) \otimes |k\rangle \langle k|, \\
    &\exp\left( \sum_{\alpha} (A)_{\alpha} \right) = \prod_{\alpha} \left(\exp(A)\right)_{\alpha},
\end{aligned}
\end{equation}
we obtain
\begin{equation}
\begin{aligned}
    U_{\text{heat}}(\tau) &= \sum_{k=0}^{N_p-1} \exp \left( i \tau \left(k-\frac{N_p}{2}\right) \sum_{\alpha=1}^{d} \left(\mathbf{H}_{0}\right)_{\alpha} \right) \otimes |k\rangle \langle k| \\
    &= \sum_{k=0}^{N_p-1} \left( \exp \left( i \tau \sum_{\alpha=1}^{d} \left(\mathbf{H}_{0}\right)_{\alpha} \right) \right)^{k-N_p/2} \otimes |k\rangle \langle k| \\
    &= \sum_{k=0}^{N_p-1} \left( \prod_{\alpha=1}^{d} \exp \left( i \tau \left(\mathbf{H}_{0}\right)_{\alpha} \right) \right)^{k-N_p/2} \otimes |k\rangle \langle k| \\
    &= \sum_{k=0}^{N_p-1} \left( \prod_{\alpha=1}^{d} \left(U_0(\tau)\right)_{\alpha} \right)^{k-N_p/2} \otimes |k\rangle \langle k| \\
    &= \sum_{k=0}^{N_p-1} \tilde{U}_{0}^{k-N_p/2}(\tau) \otimes |k\rangle \langle k|.
\end{aligned}
\end{equation}

\subsection{Proof of Equation \eqref{eqn:adv:notation}}
\label{sec:appendix:circuit:adv:notation}
By employing the first-order Lie-Trotter-Suzuki decomposition, we can approximate the time evolution operator $U_{\text{adv}}(\tau) := \exp(i \mathbf{H}_{\text{adv}} \tau)$ as follows,
\begin{equation}
\begin{aligned}
    U_{\text{adv}}(\tau) &= \exp \left( i\tau \left( \sum_{k=0}^{N_p-1} \left(k-\frac{N_p}{2}\right) \sum_{\alpha=1}^{d} |a_{\alpha}| \left(\mathbf{H}_{1}\right)_{\alpha} \otimes |k\rangle \langle k| + \sum_{\alpha=1}^{d} a_{\alpha} \left( \mathbf{H}_{2} \right)_{\alpha}  \otimes I^{\otimes n_p} \right) \right) \\
    &\approx \exp \left( i\tau \sum_{\alpha=1}^{d} a_{\alpha} \left( \mathbf{H}_{2} \right)_{\alpha}  \otimes I^{\otimes n_p} \right) \exp \left( i \tau \sum_{k=0}^{N_p-1} \left(k-\frac{N_p}{2}\right) \sum_{\alpha=1}^{d} |a_{\alpha}| \left(\mathbf{H}_{1}\right)_{\alpha} \otimes |k\rangle \langle k| \right)  \\
    &= \left( \prod_{\alpha=1}^{d} \exp \left( i a_\alpha \tau \left(\mathbf{H}_{2}\right)_{\alpha} \right) \otimes I^{\otimes n_p} \right) \sum_{k=0}^{N_p-1} \left( \prod_{\alpha=1}^{d} \exp \left( i |a_\alpha| \tau \left(\mathbf{H}_{1}\right)_{\alpha} \right) \right)^{k-N_p/2} \otimes |k\rangle \langle k| \\
    &= \left( \prod_{\alpha=1}^{d} \left( U_2(a_{\alpha}\tau) \right)_{\alpha} \otimes I^{\otimes n_p} \right) \sum_{k=0}^{N_p-1} \left( \prod_{\alpha=1}^{d} \left(U_1(|a_{\alpha}|\tau)\right)_{\alpha} \right)^{k-N_p/2} \otimes |k\rangle \langle k| \\
    &= \left( \tilde{U}_{2}(\tau) \otimes I^{\otimes n_p} \right) \sum_{k=0}^{N_p-1} \tilde{U}_{1}^{k-N_p/2}(\tau) \otimes |k\rangle \langle k|.
\end{aligned}
\end{equation}

\section{Several calculations related to the complexity analysis}
\label{sec:appendix:equations}

\subsection{Calculation of commutators}
\label{sec:appendix:equations:commutator}
    
In this section, we calculate several commutators and the upper bound of their norms, which are used in the proof of Lemma \ref{lemma:heat:error:tau} and Lemma \ref{lemma:adv:error:tau}.

Note that $\sigma_{01}^2 = \sigma_{10}^2 = 0$, $\sigma_{01}\sigma_{10} = \sigma_{00}$, $\sigma_{10}\sigma_{01} = \sigma_{11}$. Direct calculation gives
\begin{equation}
\begin{aligned}
    \left[ \sigma_{01}^{\otimes n_x}, \sigma_{10}^{\otimes n_x} \right] &= \sigma_{00}^{\otimes n_x} - \sigma_{11}^{\otimes n_x}, \\
    \left[ s_j^{-}, \sigma_{01}^{\otimes n_x} \right] &= 0, \quad \left[ s_j^{-}, \sigma_{10}^{\otimes n_x} \right] = \left(\sigma_{10}^{\otimes (n_x-1)} \otimes Z \right) \delta_{j1}, \\
    \left[ s_j^{+}, \sigma_{10}^{\otimes n_x} \right] &= 0, \quad \left[ s_j^{+}, \sigma_{01}^{\otimes n_x} \right] = \left(-\sigma_{01}^{\otimes (n_x-1)} \otimes Z \right) \delta_{j1}, \\
    \left[ s_j^{-}, s_{j^\prime}^{+} \right] &= I^{\otimes (n-j)} \otimes \left( \sigma_{00}\otimes\sigma_{11}^{\otimes(j-1)} - \sigma_{11}\otimes\sigma_{00}^{\otimes(j-1)} \right) \delta_{jj^\prime}, \\
    \left[ s_j^{-}, s_{j^\prime}^{-} \right] &= 0, \quad j \geq j^\prime >1, \quad \left[ s_{1}^{-}, s_{1}^{-} \right] = 0, \\
    \left[ s_j^{-}, s_{1}^{-} \right] &= - I^{\otimes (n_x-j)} \otimes \sigma_{01} \otimes \sigma_{10}^{\otimes (j-2)} \otimes Z, \quad j>1, \\
    \left[ s_j^{+}, s_{j^\prime}^{+} \right] &= 0, \quad j \geq j^\prime >1, \quad \left[ s_{1}^{+}, s_{1}^{+} \right] = 0, \\
    \left[ s_j^{+}, s_{1}^{+} \right] &= I^{\otimes (n_x-j)} \otimes \sigma_{10} \otimes \sigma_{01}^{\otimes (j-2)} \otimes Z, \quad j>1.
\end{aligned}
\end{equation}
As shown in Lemma \ref{lemma:changebasis} and Remark \ref{remark:changebasis:norm}, it follows that 
\begin{equation}
\begin{aligned}
    &\left\| \left[ \sigma_{01}^{\otimes n_x}, \sigma_{10}^{\otimes n_x} \right] \right\| = \left\| \left[ s_j^{-}, s_{j}^{+} \right] \right\| = 1, \\
    &\left\| \left[ s_1^{-}, \sigma_{10}^{\otimes n_x} \right] + \left[ s_1^{+}, \sigma_{01}^{\otimes n_x} \right] \right\| = \left\| \left( \sigma_{10}^{\otimes(n_x-1)} - \sigma_{01}^{\otimes (n_x-1)} \right) \otimes Z \right\| = 1, \\
    &\left\| \left[ s_1^{-}, \sigma_{10}^{\otimes n_x} \right] - \left[ s_1^{+}, \sigma_{01}^{\otimes n_x} \right] \right\| = \left\| \left( \sigma_{10}^{\otimes(n_x-1)} + \sigma_{01}^{\otimes (n_x-1)} \right) \otimes Z \right\| = 1, \\
    &\left\| \left[ s_j^{-}, s_{1}^{-} \right] + \left[ s_j^{+}, s_{1}^{+} \right] \right\| = \left\| I^{\otimes (n_x-j)} \otimes \left(\sigma_{10} \otimes \sigma_{01}^{\otimes (j-2)} - \sigma_{01} \otimes \sigma_{10}^{\otimes (j-2)} \right) \otimes Z \right\| = 1, \quad j>1, \\
    &\left\| \left[ s_{n_x}^{-}, s_{n_x}^{+} \right] - \left[ \sigma_{01}^{\otimes n_x}, \sigma_{10}^{\otimes n_x} \right] \right\| = \left\| I \otimes \left( \sigma_{11}^{\otimes (n_x-1)} - \sigma_{00}^{\otimes (n_x-1)} \right) \right\| = 1.
\end{aligned}
\end{equation}
Therefore we have the following results

\begin{equation}\label{eqn:appendix:commutator:ssigma+}
    \left\| \left[\sum_{j=1}^{n_x} ( s_{j}^{-} +  s_{j}^{+}), \sigma_{01}^{\otimes n_x} + \sigma_{10}^{\otimes n_x}\right] \right\| = \left\| \left[ s_1^{-}, \sigma_{10}^{\otimes n_x} \right] + \left[ s_1^{+}, \sigma_{01}^{\otimes n_x} \right] \right\| = 1.
\end{equation}

\begin{equation}\label{eqn:appendix:commutator:ssigma-}
    \left\| \left[ \sum_{j=1}^{n_x}( s_{j}^{-} -  s_{j}^{+}), \sigma_{01}^{\otimes n_x} - \sigma_{10}^{\otimes n_x} \right] \right\| = \left\| \left[ s_1^{-}, \sigma_{10}^{\otimes n_x} \right] + \left[ s_1^{+}, \sigma_{01}^{\otimes n_x} \right] \right\| = 1.
\end{equation}

\begin{equation}\label{eqn:appendix:commutator:s1+}
    \sum_{j=1}^{n_x} \sum_{j^\prime=j+1}^{n_x} \left\| \left[(s_{j}^{-}+s_{j}^{+}),  (s_{j^\prime}^{-}+s_{j^\prime}^{+}) \right] \right\| = \sum_{j^\prime=2}^{n_x} \left\| \left[s_{1}^{-}, s_{j^\prime}^{-} \right] + \left[s_{1}^{+}, s_{j^\prime}^{+} \right] \right\| = n_x-1.
\end{equation}

\begin{equation}\label{eqn:appendix:commutator:s1-}
    \sum_{j=1}^{n_x} \sum_{j^\prime=j+1}^{n_x} \left\| \left[(s_{j}^{-}-s_{j}^{+}),  (s_{j^\prime}^{-}-s_{j^\prime}^{+}) \right] \right\| = \sum_{j^\prime=2}^{n_x} \left\| \left[s_{1}^{-}, s_{j^\prime}^{-} \right] + \left[s_{1}^{+}, s_{j^\prime}^{+} \right] \right\| = n_x-1.
\end{equation}

\begin{equation}
\begin{aligned}
    \left\| \left[ \mathbf{H}_1, \mathbf{H}_2 \right] \right\| &= \gamma_{1}\gamma_{2} \left\| -2 \left[ \sum_{j=1}^{n_x} s_j^{-}, \sum_{j=1}^{n_x} s_j^{+} \right] + 2 \left[ \sigma_{01}^{\otimes n_x}, \sigma_{10}^{\otimes n_x} \right] \right\| \\
    &= \gamma_{1}\gamma_{2} \left\| -2 \sum_{j=1}^{n_x} \left[ s_j^{-}, s_{j}^{+} \right] + 2 \left[ \sigma_{01}^{\otimes n_x}, \sigma_{10}^{\otimes n_x} \right] \right\| \\
    &\leq  2 \gamma_{1}\gamma_{2} \left( \sum_{j=1}^{n_x-1} \left\| \left[ s_j^{-}, s_{j}^{+} \right] \right\| + \left\| \left[ s_{n_x}^{-}, s_{n_x}^{+} \right] - \left[ \sigma_{01}^{\otimes n_x}, \sigma_{10}^{\otimes n_x} \right] \right\| \right) \\
    &= 2 \gamma_{1}\gamma_{2}  n_x.
\end{aligned}
\end{equation}

\begin{equation}\label{eqn:appendix:commutator:H1H2}
\begin{aligned}
    & \left\| \left[ \mathbf{A}_{1}\otimes D_{\eta}, \mathbf{A}_{2}\otimes I^{\otimes n_p} \right] \right\| \\
    =&\left\| \left[ \sum_{k=0}^{N_p-1} \left( k - \frac{N_p}{2}\right) \sum_{\alpha=1}^{d} |a_{\alpha}| \left( \mathbf{H}_{1} \right)_{\alpha} \otimes \ket{k}\bra{k}, \sum_{\alpha=1}^{d} a_{\alpha} \left( \mathbf{H}_{2} \right)_{\alpha}  \otimes I^{\otimes n_p} \right] \right\| \\
    =& \left\| \sum_{k=0}^{N_p-1} \left( k - \frac{N_p}{2}\right) \left[ \sum_{\alpha=1}^{d} |a_{\alpha}| \left( \mathbf{H}_{1} \right)_{\alpha}, \sum_{\alpha=1}^{d} a_{\alpha} \left( \mathbf{H}_{2} \right)_{\alpha} \right] \otimes \ket{k}\bra{k} \right\| \\
    =& \left\| \sum_{k=0}^{N_p-1} \left( k - \frac{N_p}{2}\right) \sum_{\alpha=1}^{d} |a_{\alpha}| a_{\alpha} \left[ \left( \mathbf{H}_{1} \right)_{\alpha}, \left( \mathbf{H}_{2} \right)_{\alpha} \right] \otimes \ket{k}\bra{k} \right\| \\
    \leq & \frac{N_p \sum_{\alpha=1}^{d} a_{\alpha}^2}{2} \left\| \left[ \mathbf{H}_1, \mathbf{H}_2 \right] \right\| \\
    \leq & N_p \gamma_{1}\gamma_{2} n_x \sum_{\alpha=1}^{d} a_{\alpha}^2.
\end{aligned}
\end{equation}

\subsection{Approximation of operators}
\label{sec:appendix:equations:operator}

\begin{equation}\label{eqn:appendix:operator:heat}
\begin{aligned}
    &\quad \left\| U_{\text{heat}}(\tau) - V_{\text{heat}}(\tau) \right\| \\
    &= \left\| \sum_{k=0}^{N_p-1} \left[ \left( \prod_{\alpha=1}^{d} \left( U_0(\tau) \right)_{\alpha} \right)^{k-N_p/2} - \left( \prod_{\alpha=1}^{d} \left( V_0(\tau) \right)_{\alpha} \right)^{k-N_p/2} \right] \otimes |k\rangle \langle k|  \right\| \\
    &= \max_{0\leq k \leq N_p-1} \left\| \left( \prod_{\alpha=1}^{d} \left( U_0(\tau) \right)_{\alpha} \right)^{k-N_p/2} - \left( \prod_{\alpha=1}^{d} \left( V_0(\tau) \right)_{\alpha} \right)^{k-N_p/2} \right\| \\
    &\leq \max_{0\leq k \leq N_p-1} \left| k - \frac{N_p}{2} \right| \left\| \prod_{\alpha=1}^{d} \left( U_0(\tau) \right)_{\alpha} - \prod_{\alpha=1}^{d} \left( V_0(\tau) \right)_{\alpha} \right\|   \\
    &\leq \frac{d N_p}{2} \left\| U_0(\tau) - V_0(\tau) \right\| \\
    &\leq \frac{dN_p \gamma_{0}^2 \tau^2 (n_x-1)}{4}.
\end{aligned}
\end{equation}

\begin{equation}\label{eqn:appendix:operator:U1V1}
\begin{aligned}
    \left\| U_{1}(\tau) - V_{1}(\tau) \right\| &\leq \left\| U_{1}(\tau) - U_{1}^{(1)}(\tau) U_{1}^{(2)}(\tau) \right\| + \left\| U_{1}^{(1)}(\tau) U_{1}^{(2)}(\tau) - V_{1}(\tau) \right\| \\
    &\leq \frac{\gamma_{1}^2\tau^2}{2} \left\| \left[ \sum_{j=1}^{n_x} ( s_{j}^{-} +  s_{j}^{+}), \sigma_{01}^{\otimes n_x} + \sigma_{10}^{\otimes n_x} \right] \right\| + \left\| U_{1}^{(2)}(\tau) - V_{1}^{(2)}(\tau) \right\| \\
    &\leq \frac{\gamma_{1}^2 \tau^2}{2} + \frac{\gamma_{1}^2 \tau^2 (n_x-1)}{2} = \frac{\gamma_{1}^2 \tau^2 n_x}{2},
\end{aligned}
\end{equation}

\begin{equation}\label{eqn:appendix:operator:U2V2}
\begin{aligned}
    \left\| U_{2}(\tau) - V_{2}(\tau) \right\| &\leq \left\| U_{2}(\tau) - U_{2}^{(1)}(\tau) U_{2}^{(2)}(\tau) \right\| + \left\| U_{2}^{(1)}(\tau) U_{2}^{(2)}(\tau) - V_{2}(\tau) \right\| \\
    &\leq \frac{\gamma_{2}^2\tau^2}{2}  \left\| \left[ \sum_{j=1}^{n_x}( s_{j}^{-} -  s_{j}^{+}), \sigma_{01}^{\otimes n_x} - \sigma_{10}^{\otimes n_x} \right] \right\| + \left\| U_{2}^{(2)}(\tau) - V_{2}^{(2)}(\tau) \right\| \\
    &\leq \frac{\gamma_{2}^2 \tau^2}{2} + \frac{\gamma_{2}^2 \tau^2 (n_x-1)}{2} = \frac{\gamma_{2}^2 \tau^2 n_x}{2},
\end{aligned}
\end{equation}
where the results from Equation \eqref{eqn:appendix:commutator:ssigma+}, Equation \eqref{eqn:appendix:commutator:ssigma-} and Equation \eqref{eqn:heat:error:U0V0} are applied.

\end{document}